%% file: main.tex
\documentclass[letterpaper, 11pt]{article}

\usepackage{amsfonts}
\usepackage{amsmath}

\usepackage{amsthm}
\usepackage{thmtools}
\usepackage{amssymb}
\usepackage{mathtools}
\usepackage{thm-restate}
\usepackage{url}
\usepackage{authblk}
\usepackage[ruled,vlined]{algorithm2e}
\include{pythonlisting}
\usepackage{algpseudocode}
\usepackage{environ}

\usepackage{geometry}
\usepackage{xspace,fullpage}

\input{macros.tex}

\def\mc{\textsc{Max-Cut}}
\def\lmc{\textsc{Max-Cut}}
\def\llmc{\textsc{Local-Max-Cut}}
\def\nmc{\textsc{Node-Max-Cut}}
\def\lnmc{\textsc{Node-Max-Cut}}
\def\llnmc{\textsc{Local-Node-Max-Cut}}
\def\cf{\textsc{Circuit-Flip}}

\def\algo{\textsc{BridgeGaps}}

\newcommand{\thl}[1]{\textcolor{red}{#1}}

\newtheorem{theorem}{Theorem}
\newtheorem{lemma}{Lemma}
\newtheorem{proposition}[theorem]{Proposition}
\theoremstyle{remark}
\newtheorem*{remark}{Remark}
\usepackage {tikz}
\usetikzlibrary {positioning}

\theoremstyle{definition}
\newtheorem*{definition}{Definition}
\def\hf{\selectfont\sffamily\bfseries}
\newcommand{\paragr}[1]{\smallskip\noindent{\hf #1}}

\title{\bf Node Max-Cut and Computing Equilibria\\ in Linear  Weighted Congestion Games.}

\author[1]{Dimitris Fotakis}
\author[2]{Vardis Kandiros}
\author[1]{Thanasis Lianeas}
\author[1]{Nikos Mouzakis}
\author[1]{Panagiotis Patsilinakos}
\author[3]{Stratis Skoulakis}
\affil[1]{ National Technical University of Athens}
\affil[ ]{ {\textit{\{fotakis,lianeas,nmouzakis,patsilinak\}@corelab.ntua.gr}}}

\affil[2]{Massachussets Institute of Technology}
\affil[ ]{ \textit{kandiros@mit.edu}}

\affil[3]{Singapore Institute of Technology and Design}
\affil[ ]{ \textit{stratis.skoulakis@gmail.com}}

\begin{document}
\date{}
\maketitle
\begin{abstract}
In this work, we seek a more refined understanding of the complexity of local optimum computation for Max-Cut and pure Nash equilibrium (PNE) computation for congestion games with weighted players and linear latency functions. We show that computing a PNE of linear weighted congestion games is PLS-complete either for very restricted strategy spaces, namely when player strategies are paths on a series-parallel network with a single origin and destination, or for very restricted latency functions, namely when the latency on each resource is equal to the congestion. Our results reveal a remarkable gap regarding the complexity of PNE in congestion games with weighted and unweighted players, since in case of unweighted players, a PNE can be easily computed by either a simple greedy algorithm (for series-parallel networks) or any better response dynamics (when the latency is equal to the congestion). For the latter of the results above, we need to show first that computing a local optimum of a natural restriction of Max-Cut, which we call \emph{Node-Max-Cut}, is PLS-complete. In Node-Max-Cut, the input graph is vertex-weighted and the weight of each edge is equal to the product of the weights of its endpoints. Due to the very restricted nature of Node-Max-Cut, the reduction requires a careful combination of new gadgets with ideas and techniques from previous work. We also show how to compute efficiently a $(1+\eps)$-approximate equilibrium for Node-Max-Cut, if the number of different vertex weights is constant. 
\end{abstract}

\setcounter{page}{0}

\input{intro_new}

\input{prelim_new}
\input{wCGsoverview_new}
\input{apxAlgo}

\input{node-max-cut}

\input{conclusions.tex}

\newpage
\bibliography{references}

\clearpage
\appendix

\section{The Proofs of the Theorems of Section \ref{sec:congestion}}\label{s:CG}
\input{sepaPLS.tex}

\input{multiComPLS.tex}

\input{Appendix_Algo.tex}

\input{Appendix_Preliminaries.tex}

\input{Appendix_Leverage.tex}
\input{Appendix_ComputeCircuit.tex}

\input{Appendix_EqualityGadget.tex}
\input{Appendix_CopyGadget.tex}

\input{Appendix_ComparisonGadget.tex}


\end{document}

%% file: macros.tex
\newcommand{\lp}{\left}
\newcommand{\rp}{\right}
\newcommand{\mrm}[1]{\mathrm{#1}}

\newcommand{\poly}{\mrm{poly}}

\def\eps{\varepsilon}

\newlength\myindent
\newcommand{\repeattheorem}[1]{%
  \begingroup
  \renewcommand{\thetheorem}{\ref{#1}}%
  \expandafter\expandafter\expandafter\theorem
  \csname reptheorem@#1\endcsname
  \endtheorem
  \endgroup
}
\NewEnviron{reptheorem}[1]{%
  \global\expandafter\xdef\csname reptheorem@#1\endcsname{%
    \unexpanded\expandafter{\BODY}%
  }%
  \expandafter\theorem\BODY\unskip\label{#1}\endtheorem
}

\newcommand{\repeatlemma}[1]{%
  \begingroup
  \renewcommand{\thelemma}{\ref{#1}}%
  \expandafter\expandafter\expandafter\protect\lemma
  \csname replemma@#1\endcsname
  \protect\endlemma
  \endgroup
}

\NewEnviron{replemma}[1]{%
  \global\expandafter\xdef\csname replemma@#1\endcsname{%
    \unexpanded\expandafter{\BODY}%
  }%
  \expandafter\lemma\BODY\unskip\label{#1}\endlemma
}

\makeatletter
\newcommand*{\inlineequation}[2][]{%
  \begingroup
    \refstepcounter{equation}%
    \ifx\\#1\\%
    \else
      \label{#1}%
    \fi
    \relpenalty=10000 %
    \binoppenalty=10000 %
    \ensuremath{%
      #2%
    }%
    \qquad \@eqnnum
  \endgroup
}
\makeatother

%% file: intro_new.tex
\section{Introduction}
\label{s:intro}

Motivated by the remarkable success of local search in combinatorial optimization, Johnson et al. introduced \cite{JPY88} the complexity class Polynomial Local Search (PLS), consisting of local search problems with polynomially verifiable local optimality. PLS includes many natural complete problems (see e.g., \cite[App.~C]{MAK07}), with \cf~\cite{JPY88} and \mc~\cite{SY91} among the best known ones, and lays the foundation for a principled study of the complexity of local optima computation. In the last 15 years, a significant volume of research on PLS-completeness was motivated by the problem of computing a pure Nash equilibrium of potential games (see e.g., \cite{ARV08,SV08,GS10} and the references therein), where any improving deviation by a single player decreases a potential function and its local optima correspond to pure Nash equilibria \cite{MS96}. 

Computing a local optimum of \mc~under the {\sc flip} neighborhood (a.k.a. \llmc) has been one of the most widely studied problems in PLS. Given an edge-weighed graph, a cut is locally optimal if we cannot increase its weight by moving a vertex from one side of the cut to the other. Since its PLS-completeness proof by Sch\"affer and Yannakakis \cite{SY91}, researchers have shown that \llmc~remains PLS-complete for graphs with maximum degree five \cite{ET11}, is polynomially solvable for cubic graphs \cite{Pol95}, and its smoothed complexity is either polynomial in complete \cite{AN17} and sparse \cite{ET11} graphs, or almost polynomial in general graphs \cite{CGVYZ19,ER17}. Moreover, due to its simplicity and versatility, \mc~has been widely used in PLS reductions (see e.g., \cite{ARV08,GS10,SV08}). \llmc~can also be cast as a game, where each vertex aims to maximize the total weight of its incident edges that cross the cut. Cut games are potential games (the value of the cut is the potential function), which has motivated research on efficient computation of approximate equilibria for \llmc~\cite{BCK10,CFGS15}. To the best of our knowledge, apart from the work on the smoothed complexity of \llmc~(and may be that \llmc~is P-complete for unweighted graphs \cite[Theorem~4.5]{SY91}), there has not been any research on whether (and to which extent) additional structure on edge weights affects hardness of \llmc. 

A closely related research direction deals with the complexity of computing a pure Nash equilibrium (equilibrium or PNE, for brevity) of \emph{congestion games} \cite{Ros73}, a typical example of potential games \cite{MS96} and among the most widely studied classes of games in Algorithmic Game Theory (see e.g., \cite{Fot15} for a brief account of previous work). In congestion games (or CGs, for brevity), a finite set of players compete over a finite set of resources. Strategies are resource subsets and players aim to minimize the total cost of the resources in their strategies. Each resource $e$ is associated with a (non-negative and non-decreasing) latency function, which determines the cost of using $e$ as a function of $e$'s \emph{congestion} (i.e., the number of players including $e$ in their strategy). Researchers have extensively studied the properties of special cases and variants of CGs. Most relevant to this work are \emph{symmetric} (resp. \emph{asymmetric}) CGs, where players share the same strategy set (resp. have different strategy sets), \emph{network} CGs, where strategies correspond to paths in an underlying network, and \emph{weighted} CGs, where player contribute to the congestion with a different weight. 

Fabrikant et al. \cite{FPT04} proved that computing a PNE of asymmetric network CGs or symmetric CGs is PLS-complete, and that it reduces to min-cost-flow for symmetric network CGs. About the same time \cite{FKS05,PS06} proved that weighted congestion games admit a (weighted) potential function, and thus a PNE, if the latency functions are either affine or exponential (and \cite{HK10,HKM11} proved that in a certain sense, this restriction is necessary). Subsequently, Ackermann et al. \cite{ARV08} characterized the strategy sets of CGs that guarantee efficient equilibrium computation. They also used a variant of \llmc, called \emph{threshold games}, to simplify the PLS-completeness proof of \cite{FPT04} and to show that computing a PNE of asymmetric network CGs with (exponentially steep) linear latencies is PLS-complete. 

On the other hand, the complexity of equilibrium computation for weighted CGs is not well understood. All the hardness results above carry over to weighted CGs, since they generalize standard CGs (where the players have unit weight). But on the positive side, we only know how to efficiently compute a PNE for weighted CGs on parallel links with general latencies \cite{FKKMS09} and for weighted CGs on parallel links with identity latency functions and asymmetric strategies \cite{GLMM04}. Despite the significant interest in (exact or approximate) equilibrium computation for CGs (see e.g., \cite{CFGS11,CFGS15,KS17} and the references therein), we do not understand how (and to which extent) the complexity of equilibrium computation is affected by player weights. This is especially true for weighted CGs with linear latencies, which admit a potential function and their equilibrium computation is in PLS. 

\paragr{Contributions.}
We contribute to both research directions outlined above. In a nutshell, we show that equilibrium computation in linear weighted CGs is significantly harder than for standard CGs, in the sense that it is PLS-complete either for very restricted strategy spaces, namely when player strategies are paths on a series-parallel network with a single origin and destination, or for very restricted latency functions, namely when resource costs are equal to the congestion. Our main step towards proving the latter result is to show that computing a local optimum of \nmc, a natural and interesting restriction of \mc~where the weight of each edge is the product of the weights of its endpoints, is PLS-complete. 

More specifically, using a tight reduction from \llmc, we first show, in Section~\ref{sec:sepa}, that equilibrium computation for linear weighted CGs on series-parallel networks with a single origin and destination is PLS-complete (Theorem~\ref{thm:sepaPLShard}). The reduction results in games where both the player weights and the latency slopes are exponential. Our result reveals a remarkable gap between weighted and standard CGs regarding the complexity of equilibrium computation, since for standard CGs on series-parallel networks with a single origin and destination, a PNE can be computed by a simple greedy algorithm \cite{FKS05b}. 

Aiming at a deeper understanding of how different player weights affect the complexity of equilibrium computation in CGs, we show, in Section~\ref{sec:identity}, that computing a PNE of weighted network CGs with asymmetric player strategies and identity latency functions is PLS-complete (Theorem~\ref{thm:mutiPLS}). Again the gap to standard CGs is remarkable, since for standard CGs with identity latency functions, any better response dynamics converges to a PNE in polynomial time. In the reduction of Theorem~\ref{thm:mutiPLS}, \nmc~plays a role similar to that of threshold games in \cite[Sec.~4]{ARV08}. The choice of \nmc~seems necessary, in the sense that known PLS reductions, starting from {\sc Not-All-Equal Satisfiability} \cite{FPT04} or \llmc~\cite{ARV08}, show that equilibrium computation is hard due to the interaction of players on different resources (where latencies simulate the edge / clause weights), while in our setting, equilibrium computation is hard due to the player weights, which are the same for all resources in a player's strategy.

\nmc~is a natural restriction of \mc~and settling the complexity of its local optima computation may be of independent interest, both conceptually and technically. \nmc~coincides with the restriction of \mc~shown (weakly) NP-complete on complete graphs in the seminal paper of Karp \cite{Karp72}, while a significant generalization of \nmc~with polynomial weights was shown P-complete in \cite{SY91}. 

A major part of our technical effort concerns reducing \cf~to \nmc, thus showing that computing a local optimum of \nmc~is PLS-complete (Section~\ref{sec:nMCoverview}).  
Since \nmc~is a very restricted special case of \mc~ 
we have to start from a PLS-complete problem lying before \llmc~on the ``reduction paths'' of PLS. The reduction is technically involved, due to the very restricted nature of the problem. In \nmc, every vertex contributes to the cut value of its neighbors with the same weight, and differentiation comes only as a result of the different total weight in the neighborhood of each vertex. To deal with this restriction, we combine some new carefully contstructed gadgets with the gadgets used by Sch\"affer and Yannakakis \cite{SY91}, Els\"asser and Tscheuschner \cite{ET11}. and Gairing and Savani \cite{GS10}. In general, as a very resctricted special case of \mc, \nmc~is a natural and convenient starting point for future PLS reductions, especially when one wants to show hardness of equilibrium computation for restricted classes of games that admit weighted potential functions (e.g., as that in \cite{FKS05}). So, our results may serve as a first step towards a better understanding of the complexity of (exact or approximate) equilibrium computation for weighted potential games. 

We also show that a $(1+\eps)$-approximate equilibrium for \nmc, where no vertex can switch sides and increase the weight of its neighbors across the cut by a factor larger than $1+\eps$, can be computed in time exponential in the number of different weights (see Theorem~\ref{thm:equilibrium} for a precise statement). Thus, we can efficiently compute a $(1+\eps)$-approximate equilibrium for \nmc, for any $\eps > 0$, if the number of different vertex weights is constant. Since similar results are not known for \mc, we believe that Theorem~\ref{thm:equilibrium} may indicate that approximate equilibrium computation for \nmc~may not be as hard as for \mc. An interesting direction for further research is to investigate (i) the quality of efficiently computable approximate equilibria for \nmc; and (ii) the smoothed complexity of its local optima.

\paragr{Related Work.}
Existence and efficient computation of (exact or approximate) equilibria for weighted congestion games have received significant research attention. We briefly discuss here some of the most relevant previous work. There has been significant research interest in the convergence rate of best response dynamics for weighted congestion games (see e.g., \cite{CS11,CF19,EKM03,FM09,G04}). Gairing et al.~\cite{GLMM04} presented a polynomial algorithm for computing a PNE for load balancing games on restricted parallel links.  Caragiannis et al.~\cite{CFGS15} established existence and presented efficient algorithms for computing approximate PNE in weighted CGs with polynomial latencies (see also \cite{FGKS17,GNS18}).  

Bhalgat et al. \cite{BCK10} presented an efficient algorithm for computing a $(3+\eps)$-approximate equilibrium in \mc~ games, for any $\eps > 0$. The approximation guarantee was improved to $2+\eps$ in \cite{CFGS15}. We highlight that the notion of approximate equilibrium in cut games is much stronger than the notion of approximate local optimum of \mc, since the former requires that no vertex can significantly improve the total weight of its incidence edges that cross the cut (as e.g., in \cite{BCK10,CFGS15}), while the latter simply requires that the total weight of the cut cannot be significantly improved (as e.g., in \cite{CFGS15}). 

Johnson et al.~\cite{JPY88} introduced the complexity class PLS and proved that \cf~is PLS-complete. Subsequently, Sch\"affer and Yannakakis \cite{SY91} proved that \mc~is PLS-complete. From a technical viewpoint, our work is close to previous work by Els\"asser and Tscheuschner \cite{ET11} and Gairing and Savani \cite{GS10}, where they show that \llmc~in graphs of maximum degree five \cite{ET11} and computing a PNE for hedonic games \cite{GS10} are PLS-complete, and by Ackermann et al. \cite{ARV08}, where they reduce \llmc~to computing a PNE in network congestion games.

%% file: prelim_new.tex
\section{Basic Definitions and Notation}
\label{sec:prelim}

\paragr{Polynomial-Time Local Search (PLS).}
A \emph{polynomial-time local search} (PLS) problem $L$ \cite[Sec.~2]{JPY88} is specified by a (polynomially recognizable) set of instances $I_L$, a set $S_L(x)$ of feasible solutions for each instance $x \in I_L$, with $|s| = O(\poly(|x|)$ for every solution $s \in S_L(x)$, an objective function $f_L(s, x)$ that maps each solution $s \in S_L(x)$ to its value in instance $x$, and a \emph{neighborhood} $N_L(s, x) \subseteq S_L(x)$ of feasible solutions for each $s \in S_L(x)$. Moreover, there are three polynomial-time algorithms that for any given instance $x \in I_L$: (i) the first generates an initial solution $s_0 \in S_L(x)$; (ii) the second determines whether a given $s$ is a feasible solution and (if $s \in S_L(x)$) computes its objective value $f_L(s, x)$; and (iii) the third returns either that $s$ is \emph{locally optimal} or a feasible solution $s' \in N_L(s, x)$ with better objective value than $s$. If $L$ is a maximization (resp. minimization) problem, a solution $s$ is locally optimal if for all $s' \in N_L(s, x)$, $f_L(s, x) \geq f_L(s', x)$ (resp. $f_L(s, x) \leq f_L(s', x)$). If $s$ is not locally optimal, the third algorithm returns a solution $s' \in N_L(s, x)$ with $f(s, x) < f(s', x)$ (resp. $f(s, x) > f(s', x)$). The complexity class PLS consists of all polynomial-time local search problems. By abusing the terminology, we always refer to polynomial-time local search problem simply as local search problems. 

\paragr{PLS Reductions and Completeness.} 
A local search problem $L$ is PLS-\emph{reducible} to a local search problem
$L'$, if there are polynomial-time algorithms $\phi_1$ and $\phi_2$ such that (i) $\phi_1$ maps any instance $x \in I_L$ of $L$ to an instance $\phi_1(x) \in I_{L'}$ of $L'$; (ii) $\phi_2$ maps any (solution $s'$ of instance $\phi_1(x)$, instance $x$) pair, with $s' \in S_{L'}(\phi_1(x))$, to a solution $s \in S_L(x)$; and (iii) for every instance $x \in I_L$, if $s'$ is locally optimal for $\phi_1(x)$, then $\phi_2(s', x)$ is locally optimal for $x$. 

By definition, if a local search problem $L$ is PLS-reducible to a local search problem $L'$, a polynomial-time algorithm that computes a local optimum of $L'$ implies a polynomial time algorithm that computes a local optimum of $L$. Moreover,  a PLS-reduction is transitive. As usual, a local search problem $Q$ is PLS-\emph{complete}, if $Q \in \text{PLS}$ and any local search problem $L \in \text{PLS}$ is PLS-reducible to $Q$. 

\paragr{Max-Cut and Node-Max-Cut.}
An instance of \mc~consists of an undirected edge-weighted graph $G(V, E)$, where $V$ is the set of vertices and $E$ is the set of edges. Each edge $e$ is associated with a positive weight $w_e$. A cut of $G$ is a vertex partition $(S, V \setminus S)$, with $\emptyset \neq S \neq V$. We usually identify a cut with one of its sides (e.g., $S$). We denote $\delta(S) = \{ \{u, v\} \in E : u \in S \land v \not\in S \}$ the set of edges that cross the cut $S$. The weight (or the value) of a cut $S$, denoted $w(S)$, is $w(S) = \sum_{e \in \delta(S)} w_e$. In \mc, the goal is to compute an optimal cut $S^\ast$ of maximum value $w(S^\ast)$. 

In \nmc, each vertex $v$ is associated with a positive weight $w_v$ and the weight of each edge $e = \{ u, v \}$ is $w_e = w_u w_v$, i.e. equal to the product of the weights of $e$'s endpoints. Again the goal is to compute a cut $S^\ast$ of maximum value $w(S^\ast)$. As optimization problems, both \mc~and \nmc~are NP-complete \cite{Karp72}. 

In this work, we study \mc~and \nmc~as local search problems under the {\sc flip} neighborhood. Then, they are referred to as \llmc~and \llnmc. The neighborhood $N(S)$ of a cut $(S, V \setminus S)$ consists of all cuts $(S', V \setminus S')$ where $S$ and $S'$ differ by a single vertex. Namely, the cut $S'$ is obtained from $S$ by moving a vertex from one side of the cut to the other. A cut $S$ is locally optimal if for all $S' \in N(S)$, $w(S) \geq w(S')$. In \llmc~(resp. \llnmc), given an edge-weighted (resp. vertex-weighted) graph, the goal is to compute a locally optimal cut. Clearly, both \mc{} and \nmc{} belong to PLS. In the following, we abuse the terminology and refer to \llmc~and \llnmc~as \mc{} and \nmc{}, for brevity, unless we need to distinguish between the optimization and the local search problem. 

%

\paragr{Weighted Congestion Games.} 
A \emph{weighted congestion game} $\mathcal{G}$ consists of $n$ players, where each player $i$ is associated with a positive weight $w_i$, a set of resources $E$, where each resource $e$ is associated with a non-decreasing latency function $\ell_e:\mathbb{R}_{\geq0}\to \mathbb{R}_{\geq0}$, and a non-empty strategy set $\Sigma_i \subseteq 2^E$ for each player $i$. A game is linear if $\ell_e(x) = a_e x+b_e$, for some $a_e, b_e \geq 0$, for all $e \in E$. The identity latency function is $\ell(x) =x$. The player strategies are \emph{symmetric}, if all players share the same strategy set $\Sigma$, and \emph{asymmetric}, otherwise. 

We focus on \emph{network} weighted congestion games, where the resources $E$ correspond to the edges of an underlying network $G(V, E)$ and the player strategies are paths on $G$. A network game is \emph{single-commodity}, if $G$ has an origin $o$ and a destination $d$ and the player strategies are all (simple) $o-d$ paths. A network game is \emph{multi-commodity}, if $G$ has an origin $o_i$ and a destination $d_i$ for each player $i$, and $i$'s strategy set $\Sigma_i$ consists of all (simple) $o_i - d_i$ paths. 
A single-commodity network $G(V, E)$ is \emph{series-parallel}, if it either consists of a single edge $(o, d)$ or can be obtained from two series-parallel networks composed either in series or in parallel (see e.g., \cite{VTL82} for details on composition and recognition of series-parallel networks). 

A configuration $\vec{s} = (s_1, \ldots, s_n)$ consists of a strategy $s_i \in \Sigma_i$ for each player $i$. The congestion $s_e$ of resource $e$ in configuration $\vec{s}$ is $s_e = \sum_{i: e \in s_i} w_i$. The cost of resource $e$ in $\vec{s}$ is $\ell_e(s_e)$. The \emph{individual cost} (or \emph{cost}) $c_i(\vec{s})$ of player $i$ in configuration $\vec{s}$ is the total cost for the resources in her strategy $s_i$, i.e., $c_i(\vec{s}) = \sum_{e \in s_i} \ell_e(s_e)$. A configuration $\vec{s}$ is a \emph{pure Nash equilibrium} (equilibrium or PNE, for brevity), if for every player $i$ and every strategy $s' \in \Sigma_i$, $c_i(\vec{s}) \leq c_i(\vec{s}_{-i}, s')$ (where $(\vec{s}_{-i}, s')$ denotes the configuration obtained from $\vec{s}$ by replacing $s_i$ with $s'$). Namely, no player can improve her cost by unilaterally switching her strategy. 

\paragr{Equilibrium Computation and Local Search.} 
\cite{FKS05} shows that for linear weighted congestion games, with latencies $\ell_e(x) = a_ex+b_e$, $\Phi(\vec{s}) = \sum_{e \in E} (a_e s_e^2 + b_e s_e) + \sum_{i} w_i \sum_{e \in s_i} (a_e w_i + b_e)$ changes by $2w_i(c_i(\vec{s}) - c_i(\vec{s}_{-i}, s'))$, when a player $i$ switches from strategy $s_i$ to strategy $s'$ in $\vec{s}$. Hence, $\Phi$ is a weighted potential function, whose local optimal (wrt. single player deviations) correspond to PNE of the underlying game. 
Hence, equilibrium computation for linear weighted congestion games is in PLS. Specifically, configurations corresponds to solutions, the neighborhood $N(\vec{s})$ of a configuration $\vec{s}$ consists of all configurations $(\vec{s}_{-i}, s')$ with $s' \in \Sigma_i$, for some player $i$, and local optimality is defined wrt. the potential function $\Phi$. 

\paragr{Max-Cut and Node-Max-Cut as Games.} 
\llmc{} and \llnmc{} can be cast as cut games, where players correspond to vertices of $G(V, E)$, strategies $\Sigma = \{ 0, 1 \}$ are symmetric, and configurations $\vec{s} \in \{0, 1\}^{|V|}$ correspond to cuts, e.g., $S(\vec{s}) = \{ v \in V: s_v = 0\}$. Each player $v$ aims to maximize $w_v(\vec{s}) = \sum_{e = \{u, v\} \in E: s_u \neq s_v} w_e$, that is the total weight of her incident edges that cross the cut. For \nmc{}, this becomes $w_v(\vec{s}) = \sum_{u: \{u, v\} \in E \land s_u \neq s_v} w_u$, i.e., $v$ aims to maximize the total weight of her neighbors across the cut. A cut $\vec{s}$ is a PNE if for all players $v$, $w_v(\vec{s}) \geq w_v(\vec{s}_{-i}, 1-s_v)$. Equilibrium computation for cut games is equivalent to local optimum computation, and thus, is in PLS. 

A cut $\vec{s}$ is a $(1+\eps)$-approximate equilibrium, for some $\eps > 0$,  if for all players $v$, $(1+\eps)w_v(\vec{s}) \geq w_v(\vec{s}_{-i}, 1-s_v)$. Note that the notion of $(1+\eps)$-approximate equilibrium is stronger than the notion of $(1+\eps)$-approximate local optimum, i.e., a cut $S$ such that for all $S' \in N(S)$, $(1+\eps)w(S) \geq w(S')$ (see also the discussion in \cite{CFGS15}). 

%% file: wCGsoverview_new.tex
\section{Hardness of Computing Equilibria in Weighted Congestion Games}
\label{sec:congestion}

We next show that computing a PNE in weighted congestion games with linear latencies is PLS-complete either for single-commodity series-parallel networks or for multi-commodity networks with identity latency functions. Missing technical details can be found in Appendix~\ref{s:CG}.

\subsection{Weighted Congestion Games on Series-Parallel Networks}
\label{sec:sepa}

\begin{theorem}\label{thm:sepaPLShard}
Computing a pure Nash equilibrium in weighted congestion games on single-commodity series-parallel networks with linear latency functions is PLS-complete.
\end{theorem}

\begin{proof}[Proof sketch]
Membership in PLS follows from the potential function argument of \cite{FKS05}. To show hardness, we present a reduction from \lmc{} (see also Appendix~\ref{sec:singlePLS}). 

Let $H(V,A)$ be an instance of \llmc{} with $n$ vertices and $m$ edges. Based on $H$, we construct a weighted congestion game on a single-commodity series-parallel network $G$ with $3n$ players, where for every $i \in [n]$, there are three players with weight $w_i = 16^i$. 
Network $G$ is a parallel composition of two identical copies of a simpler series-parallel network. We refer to these copies as $G_1$ and $G_2$. Each of $G_1$ and $G_2$ is a series composition of $m$ simple series-parallel networks $F_{ij}$, each corresponding to an edge $\{i,j\}\in A$. Network $F_{ij}$ is depicted in Figure~\ref{fig:sepaPLS}, where $D$ is assumed to be a constant chosen (polynomially) large enough. An example of the entire network $G$ is shown in Figure~\ref{fig:sepaExample}. 

\begin{figure}[t]
    \centering
    \includegraphics[width=0.75\textwidth]{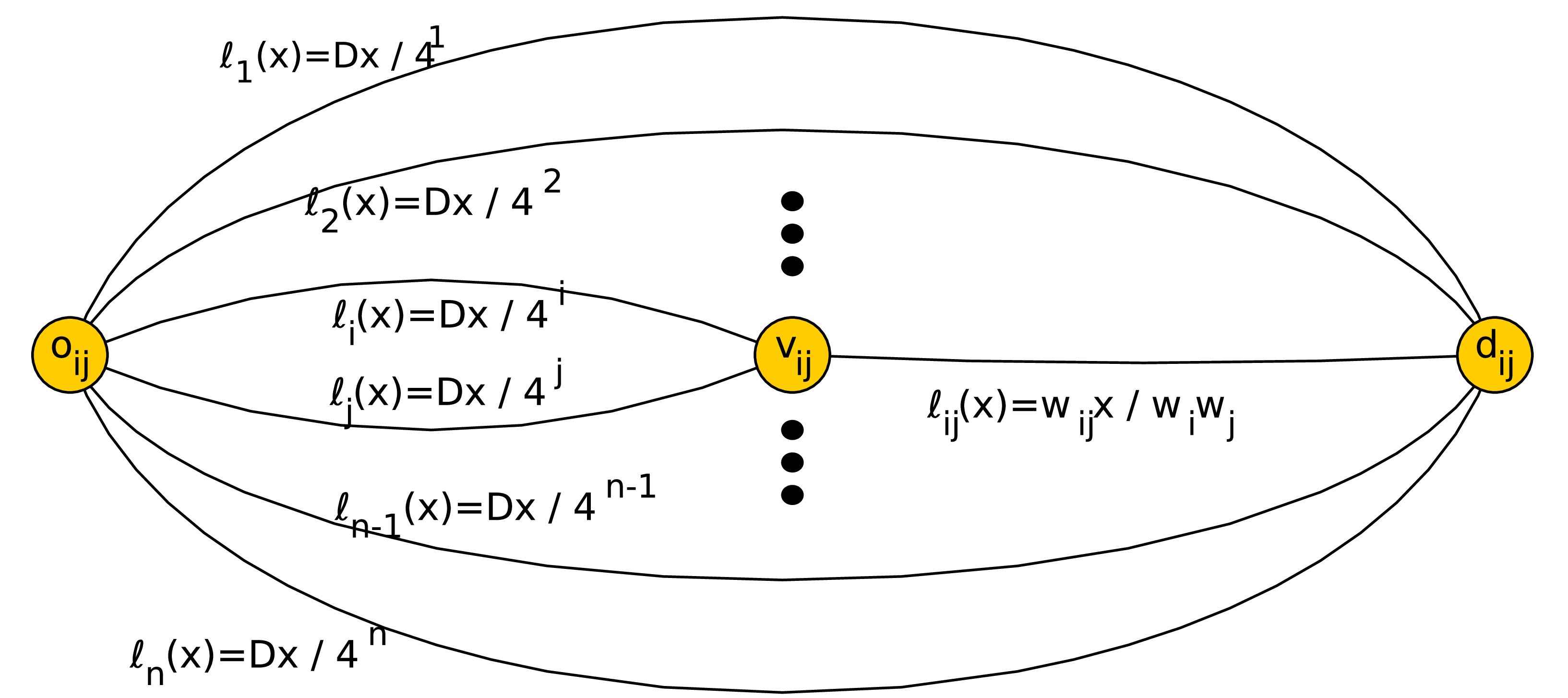}
    \caption{The series-parallel network $F_{ij}$ that corresponds to edge $\{i,j\}\in A$.}
    \label{fig:sepaPLS}
\end{figure}

\begin{figure}
    \centering
    \includegraphics[width=\textwidth]{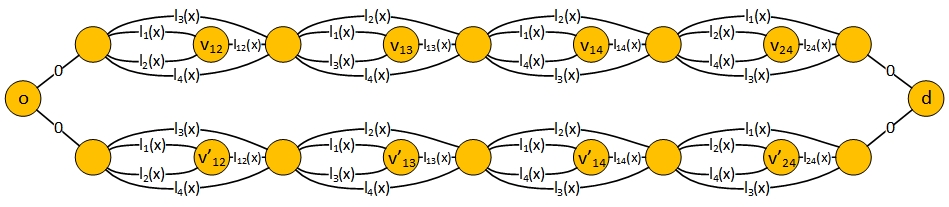}
    \caption{An example of the network $G$ constructed in the proof of Theorem~\ref{thm:sepaPLShard} for graph $H(V,A)$, with $V=\{1,2,3,4\}$ and $A=\{\{1,2\},\{1,3\},\{1,4\},\{2,4\}\}$. $G$ is a parallel composition of two parts, each consisting of the smaller networks $F_{12}$, $F_{13}$, $F_{14}$ and $F_{24}$ (see also Figure~\ref{fig:sepaPLS}) connected in series.}
    \label{fig:sepaExample}
\end{figure}

In each of $G_1$ and $G_2$, there is a unique path that contains all  edges with latency functions $\ell_i(x) = Dx/4^i$, for each $i\in[n]$. We refer to these paths as $p^u_i$ for $G_1$ and $p_i^l$ for $G_2$. In addition to the edges with latency $\ell_i(x)$, $p^u_i$ and $p_i^l$ include all edges with latencies $\ell_{ij}(x) = \frac{w_{ij}x}{w_iw_j} = \frac{w_{ij}x}{16^{i+j}}$, which correspond to the edges incident to vertex $i$ in $H$. 

Due to the choice of the player weights and the latency slopes, a player with weight $w_i$ must choose either $p^u_i$ or $p_i^l$ in any PNE (see also propositions~\ref{prop:PlayersToPaths}~and~\ref{prop:playersToPaths2} in Appendix~\ref{sec:singlePLS}). We can prove this claim by induction on the player weights. The players with weight $w_n = 16^n$ have a dominant strategy to choose either $p^u_n$ or $p_n^l$, since the slope of $\ell_n(x)$ is significantly smaller than the slope of any other latency $\ell_i(x)$. In fact, the slope of $\ell_n$ is so small that even if all other $3n-1$ players choose one of $p^u_n$ or $p_n^l$, a player with weight $w_n$ would prefer either $p^u_n$ or $p_n^l$ over all other paths. Therefore, we can assume that each of $p^u_n$ and $p_n^l$ are used by at least one player with weight $w_n$ in any PNE, which would increase their latency so much that no player with smaller weight would prefer them any more. The inductive argument applies the same reasoning for players with weights $w_{n-1}$, who should choose either $p^u_{n-1}$ or $p_{n-1}^l$ in any PNE, and subsequently, for players with weights $w_{n-2}, \ldots, w_{1}$. 
Hence, we conclude that for all $i \in [n]$, each of $p^u_i$ and $p_i^l$ is used by at least one player with weight $w_i$. 

Moreover, we note that two players with different weights, say $w_i$ and $w_j$, go through the same edge with latency $\ell_{ij}(x) = \frac{w_{ij}x}{w_iw_j}$ in $G$ only if the corresponding edge $\{i,j\}$ is present in $H$. The correctness of the reduction follows the fact that a player with weight $w_i$ aims to minimize her cost through edges with latencies $\ell_{ij}$ in $G$ in the same way that in the \lmc~instance, we want to minimize the weight of the edges incident to a vertex $i$ and do not cross the cut. 
Formally, we next show that a cut $S$ is locally optimal for the \lmc{} instance if and only if the configuration where for every $k \in S$, two players with weight $w_k$ use $p_k^u$ and for every $k \not\in S$, two players with weight $w_k$ use $p_k^l$ is a PNE of the weighted congestion game on $G$.

Assume an equilibrium configuration and consider a player $a$ of weight $w_k$ that uses $p_k^u$ together with another player of weight $w_k$ (if this is not the case, vertex $k$ is not included in $S$ and we apply the symmetric argument for $p_k^l$). By the equilibrium condition, the cost of player $a$ on $p_k^u$ is at most her cost on $p_k^l$, which implies that 
$$\sum_{k=1}^m\frac{2D16^k}{4^k}+\sum_{j: \{k,j\}\in A}\frac{w_{kj}(2\cdot16^k+x_j^u16^j)}{16^{k+j}} \leq \sum_{k=1}^m\frac{2D16^k}{4^k}+\sum_{j: \{k,j\}\in A}\frac{w_{kj}(2\cdot16^k+x_j^l16^j)}{16^{k+j}}\,,$$
where $x^u_j$ (resp. $x_j^l$) is either $1$ or $2$ (resp. $2$ or $1$) depending on whether, for each vertex $j$ connected to vertex $k$ in $H$, one or two players (of weight $w_j$) use $p_j^u$. Simplifying the inequality above, we obtain that: 
\begin{equation}\label{eqn:LmcToEqMain}
 \sum_{j: \{k,j\}\in A}w_{kj}(x_j^u-1) \leq 
 \sum_{j: \{k,j\}\in A}w_{kj}(x_j^l-1)
\end{equation}

Let $S=\{i\in V : x_i^u=2 \}$. By hypothesis, $k\in S$ and the left-hand side of \eqref{eqn:LmcToEqMain} corresponds to the total weight of the edges in $H$ that are incident to $k$ and do not cross the cut $S$. Similarly, the right-hand side of \eqref{eqn:LmcToEqMain} corresponds to the total weights of the edges in $H$ that are incident to $k$ and cross the cut $S$. Therefore, \eqref{eqn:LmcToEqMain} implies that we cannot increase the value of the cut $S$ by moving vertex $k$ from $S$ to $V \setminus S$. Since this or its symmetric condition holds for any vertex $k$ of $H$, the cut $(S, V\setminus S)$ is locally optimal. To conclude the proof, we argue along the same lines that any locally optimal cut of $H$ corresponds to a PNE in the weighted congestion game on $G$. 
\end{proof}

\subsection{Weighted Congestion Games with Identity Latency Functions}
\label{sec:identity}

We next prove that computing a PNE in weighted congestion games on multi-commodity networks with identity latency functions is PLS-complete. Compared to Theorem~\ref{thm:sepaPLShard}, we allow for a significantly more general strategy space, but we significantly restrict the latency functions, only allowing for the player weights to be exponentially large. 
\begin{theorem}\label{thm:mutiPLS}
Computing a pure Nash equilibrium in weighted congestion games on multi-commodity networks with identity latency functions is PLS-complete.
\end{theorem}

\begin{proof}[Proof sketch.]
We use a reduction from \nmc{}, which as we show in Theorem~\ref{t:main_max_cut}, is PLS-complete. Our construction draws ideas from \cite{ARV08}. Missing technical details can be found in Appendix~\ref{sec:multiPLS}. 

Let $H(V, A)$ be an instance of \lnmc{}. We construct a weighted congestion game on a multi-commodity network $G$ with identity latency functions $\ell_e(x) = x$ such that equilibria of the congestion game correspond to locally optimal cuts of $H$. 

At the conceptual level, each player $i$ of the congestion game corresponds to vertex $i \in V$ and has weight $w_i$ (i.e., equal to the weight of vertex $i$ in $H$). The key step is to construct a network $G$ (see also Figure~\ref{fig:multiPLS}) such that for every player $i \in [n]$, there are two paths, say $p^u_i$ and $p^l_i$, whose cost dominate the cost of any other path for player $i$. Therefore, in any equilibrium, player $i$ selects either $p^u_i$ or $p^l_i$ (which corresponds to vertex $i$ selecting one of the two sides of a cut). For every edge $\{i, j \} \in A$, paths $p^u_i$ and $p^u_j$ (resp. paths $p^l_i$ and $p^l_j$) have an edge $e^u_{ij}$ (resp. $e^l_{ij}$) in common. Intuitively, the cost of $p^u_i$ (resp. $p^l_i$) for player $i$ is determined by the set of players $j$, with $j$ connected to $i$ in $H$, that select path $p^u_j$ (resp. $p^l_j$). 

Let us consider any equilibrium configuration $\vec{s}$ of the weighted congestion game. By the discussion above, each player $i\in[n]$ selects either $p^u_i$ or $p^l_i$ in $\vec{s}$. Let $S=\{ i \in [n]: \mbox{player $i$ selects $p^u_i$ in $\vec{s}$}\}$. Applying the equilibrium condition, we next show that $S$ is a locally optimal cut. 

We let $V_i = \{ j: \{i, j\} \in A\}$ be the neighborhood of vertex $i$ in $H$. By the construction of $G$, the individual cost of a player $i$ on path $p^u_i$ (resp. $p^l_i$) in $\vec{s}$ is equal to $K + |V_i| w_i + \sum_{j \in S \cap V_i} w_j$ (resp. $K + |V_i| w_i + \sum_{j \in V_i \setminus S} w_j$), where $K$ is a large constant that depends on the network $G$ only. Therefore, for any player $i \in S$, equilibrium condition for $\vec{s}$ implies that
\[ K + |V_i| w_i + \sum_{j \in S \cap V_i} w_j \leq K + |V_i| w_i + \sum_{j \in V_i \setminus S} w_j 
\Rightarrow 
\sum_{j \in S \cap V_i} w_j \leq \sum_{j \in V_i \setminus S} w_j  \]
Multiplying both sides by $w_i$, we get that the total weight of the edges that are incident to $i$ and cross the cut $S$ is no less than the total weight of the edges that are incident to $i$ and do not cross the cut. By the same reasoning, we reach the same conclusion for any player $i \not\in S$. Therefore, the cut $(S, V \setminus S)$ is locally optimal for the \nmc{} instance $H(V, A)$. 

To conclude the proof, we argue along the same lines that any locally optimal cut $S$ for the \nmc{} instance $H(V, A)$ corresponds to an equilibrium in the network $G$, by letting a player $i$ select path $p^u_i$ if and only if $i \in S$. 
\end{proof}

%% file: apxAlgo.tex
\section{Computing Approximate Equilibria for Node-Max-Cut}
\label{s:approx-algo}

We complement our PLS-completeness proof for \nmc{}, in Section~\ref{sec:nMCoverview}, with an efficient algorithm computing $(1+\eps)$-approximate equilibria for \nmc{}, when the number of different vertex weights is a constant. We note that similar results are not known (and a similar approach fails) for \mc{}. Investigating if stronger approximation guarantees are possible for efficiently computable approximate equilibria for \nmc{} is beyond the scope of this work and an intriguing direction for further research. 

Given a vertex-weighted graph $G(V, E)$ with $n$ vertices and $m$ edges,
our algorithm, called \algo{} (Algorithm~\ref{Algorithm}), computes a $(1+\eps)^3$-approximate equilibrium for a \nmc{}, for any $\eps > 0$, in $(m/\eps)(n/\eps)^{O(D_\eps)}$ time, where $D_\eps$ is the number of different vertex weights in $G$, when the weights are rounded down to powers of $1+\eps$. We next sketch the algorithm and the proof of Theorem~\ref{thm:equilibrium}. Missing technical details can be found in Appendix~\ref{app:algo}.

For simplicity, we assume that $n/\eps$ is an integer (we use $\lceil n/\eps \rceil$ in Appendix~\ref{app:algo}) and that vertices are indexed in nondecreasing order of weight, i.e., $w_1 \leq w_2 \leq \cdots \leq w_n$. \algo{} first rounds down vertex weights to the closest power of $(1+\eps)$. Namely, each weight $w_i$ is replaced by weight $w'_i=(1+\eps)^{\lfloor\log_{1+\eps}w_i\rfloor}$. Clearly, an $(1+\eps)^2$-approximate equilibrium for the new instance $G'$ is an $(1+\eps)^3$-approximate equilibrium for the original instance $G$. The number of different weights $D_\eps$, used in the analysis, is defined wrt. the new instance $G'$. 

Then, \algo{} partitions the vertices of $G'$ into groups $g_1, g_2, \ldots$, so that the vertex weights in each group increase with the index of the group and the ratio of the maximum weight in group $g_j$ to the minimum weight in group $g_{j+1}$ is no less than $n/\eps$. This can be performed by going through the vertices, in nondecreasing order of their weights, and assign vertex $i+1$ to the same group as vertex $i$, if $w'_{i+1}/w'_i \leq n/\eps$. Otherwise, vertex $i+1$ starts a new group. The idea is that for an $(1+\eps)^2$-approximate equilibrium in $G'$, we only need to enforce the $(1+\eps)$-approximate equilibrium condition for each vertex $i$ only for $i$'s neighbors in the highest-indexed group (that includes some neighbor of $i$). To see this, let $g_j$ be the highest-indexed group that includes some neighbor of $i$ and let $\ell$ be the lowest indexed neighbor of $i$ in $g_j$. Then, the total weight of $i$'s neighbors in groups $g_1, \ldots, g_{j-1}$ is less than $\eps w'_{\ell}$. This holds because $i$ has at most $n-2$ neighbors in these groups and by definition, $w'_q \leq (\eps/n) w'_{\ell}$, for any $i$'s neighbor $q$ in groups $g_1, \ldots, g_{j-1}$. Therefore, we can ignore all neighbors of $i$ in groups $g_1, \ldots, g_{j-1}$, at the expense of one more $1+\eps$ factor in the approximate equilibrium condition (see also the proof of Lemma~\ref{lem:apxGrntAlgo}). 

Since for every vertex $i$, we need to enforce its (approximate) equilibrium condition only for $i$'s neighbors in a single group, we can scale down vertex weights in the same group uniformly (i.e., dividing all the weights in each group by the same factor), as long as we maintain the key property in the definition of groups (i.e., that the ratio of the maximum weight in group $g_j$ to the minimum weight in group $g_{j+1}$ is no less than $n/\eps$). Hence, we uniformly scale down the weights in each group so that (i) the minimum weight in group $g_1$ becomes $1$; and (ii) for each $j \geq 2$, the ratio of the maximum weight in group $g_{j-1}$ to the minimum weight in group $g_{j}$ becomes exactly $n/\eps$ (see Appendix~\ref{app:algo} for the details). This results in a new instance $G''$ where the minimum weight is $1$ and the maximum weight is $(n/\eps)^{D_\eps}$. Therefore, a $(1+\eps)$-approximate equilibrium in $G''$ can be computed, in a standard way, after at most $(m\eps)(n/\eps)^{2D_\eps}$ $\eps$-best response moves (see the proof of Lemma~\ref{lem:runTimeAlgo}). 

Putting everything together and using $\eps' = \eps/7$, so that $(1+\eps')^3 \leq 1+\eps$, for all $\eps \in (0, 1]$, we obtain the following (see Appendix~\ref{app:algo} for the formal proof). We note that the running time of \algo{} is polynomial, if $D_\eps = O(1)$ (and quasipolynomial if $D_\eps = \mathrm{poly}(\log n)$). 

\begin{theorem}\label{thm:equilibrium}
For any vertex-weighted graph $G$ with $n$ vertices and $m$ edges and any $\eps > 0$, \algo{} computes a $(1+\eps)$-approximate pure Nash equilibrium for \nmc{} on $G$ in  $(m/\eps)(n/\eps)^{O(D_\eps)}$ time, where $D_\eps$ denotes the number of different vertex weights in $G$, after rounding them down to the nearest power of $1+\eps$. 
\end{theorem}

%% file: node-max-cut.tex
\section{PLS-Completeness of \lnmc{}}
\label{sec:nMCoverview}

We outline the proof of Theorem~\ref{t:main_max_cut} and the main differences of our reduction from known PLS reductions to \mc{} \cite{ET11,GS10,SY91}. The technical details can be found in Appendix~\ref{s:MC}. 

\begin{theorem}\label{t:main_max_cut}
\llnmc{} is PLS-complete.
\end{theorem}

As discussed in Section~\ref{sec:prelim}, the local search version of \nmc{} is in PLS. To establish PLS-hardness of \nmc{}, we present a reduction from \cf{}. 

An instance of \cf~consists of a Boolean circuit $C$ with $n$ inputs and $m$ outputs (and wlog. only NOR gates). The value $C(s)$ of an $n$-bit input string $s$ is the integer corresponding to the $m$-bit output string. The neighborhood $N(s)$ of $s$ consists of all $n$-bit strings $s'$ at Hamming distance $1$ to $s$ (i.e., $s'$ is obtained from $s$ by flipping one bit of $s$). The goal is to find a locally optimal input string $s$, i.e., an $s$ with $C(s) \geq C(s')$, for all $s' \in N(s)$. \cf~was the first problem shown to be PLS-complete 
in \cite{JPY88}.

\begin{figure}[t]
\centering\includegraphics[scale = 0.5]{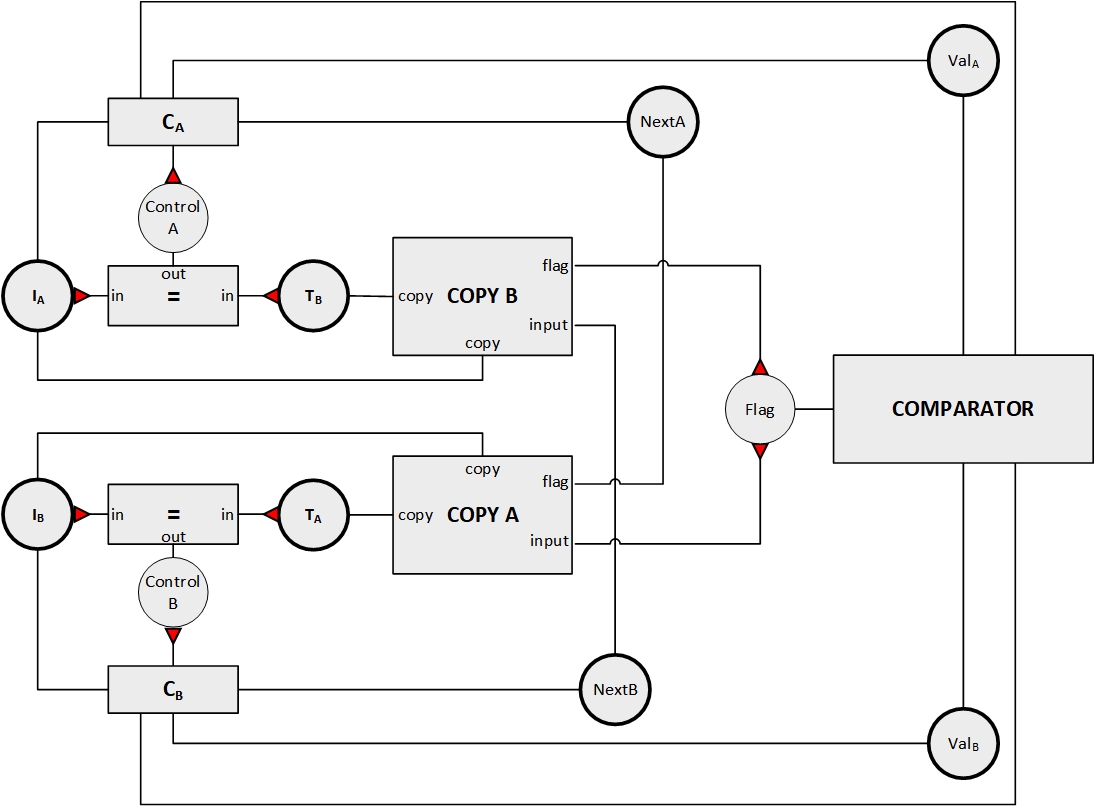}
\caption{The general structure of the \nmc{} instance constructed in the proof of Theorem~\ref{t:main_max_cut}. Rectangles denote the main gadgets, which are defined and discussed in Appendix~\ref{s:MC}, circles denote vertices that participate in multiple gadgets, and circles with bold border denote groups of $n$ such vertices. The small red triangles are used to indicate the ``information flow''. 
}
\label{f:construction}
\end{figure}

Given an instance $C$ of \cf{}, we construct below a vertex-weighted undirected graph $G(V, E)$ so that from any locally optimum cut of $G$, we can recover, in polynomial time, a locally optimal input of $C$. The graph $G$ consists of different gadgets (see Figure~\ref{f:construction}), which themselves might be regarded as smaller instances of \lnmc{}. Intuitively, each of these gadgets receives information from its ``input'' vertices, process this information, while carrying it through its internal part, and outputs the results through its ``output'' vertices. Different gadgets are glued together through their ``input'' and ``output'' vertices. 

Our construction follows the \emph{flip-flop} architecture (Figure~\ref{f:construction}), previously used e.g., in \cite{ET11,GS10,SY91}, but requires more sophisticated implementations of several gadgets, so as to conform with the very restricted weight structure of \lnmc. Next, we outline the functionality of the main gadgets and how the entire construction works (see Appendix~\ref{sec:outline}). 

Given a circuit $C$, we construct two \emph{Circuit Computing} gadgets $C_\ell$ ($\ell$ always stands for either $A$ or $B$), which are instances of \lnmc{} that simulate circuit $C$ in the following sense: Each $C_\ell$ has a set $I_\ell$ of $n$ ``input'' vertices, whose (cut) values correspond to the input string of circuit $C$, and a set $Val_\ell$ of $m$ ``output'' vertices, whose values correspond to the output string of $C$ on input $I_\ell$. There is also a set $Next\ell$ of $n$ vertices whose values correspond to a $n$-bit string in the neighborhood of $I_\ell$ of circuit value larger than that of $I_\ell$ (if the values of $Next\ell$ coincide  with the values of $I_\ell$, $I_\ell$ is locally optimal). 

The \emph{Circuit Computing} gadgets operate in two different modes, determined by  $Control_\ell$: the \emph{write mode}, when $Control_\ell = 0$, and the \emph{compute mode}, when $Control_\ell = 1$. If $C_\ell$ operates in write mode,
the input values of $I_\ell$ can be updated to the values of the complementary $Next$ set (i.e., $I_A$ is updated to $NextB$, and vice versa). When, $C_\ell$ operates in the compute mode, $C_\ell$ simulates the computation of circuit $C$ and the values of $Next\ell$ and $Val_\ell$ are updated to the corresponding values produced by $C$. Throughout the proof, we let $Real\text{-}Val(I_\ell)$ denote the output string of circuit $C$ on input $I_\ell$ (i.e., the input of $C$ takes the cut values of the vertices in $I_\ell$), and let $Real\text{-}Next(I_\ell)$ denote a neighbor of $I_\ell$ with circuit value larger than the circuit value of $I_\ell$. If $I_\ell$ is locally optimal, $Real\text{-}Next(I_\ell) = I_\ell$.

Our \emph{Circuit Computing} gadgets $C_A$ and $C_B$ are based on the gadgets of Sch\"affer and Yannakakis \cite{SY91} (see also Figure~\ref{f:CGadget} for an abstract description of them). Their detailed construction is described in Section~\ref{s:computing_gadgets} and their properties are summarized in Theorem~\ref{l:SY_gadgets}. 

The \emph{Comparator} gadget compares $Val_A$ with $Val_B$, which are intended to be $Real\text{-}Val(I_A)$ and $Real\text{-}Val(I_B)$, respectively, and outputs $1$, if $Val_A \leq Val_B$, and $0$, otherwise. The result of the \emph{Comparator} is stored in the value of the \emph{Flag} vertex. If $Flag = 1$, the \emph{Circuit Computing} gadget $C_A$ enters its write mode and the input values in $I_A$ are updated to the neighboring solution of $I_B$, currently stored in $NextB$ (everything operates symmetrically, if $Flag = 0$). Then, in the next ``cycle'', the input in $I_A$ leads $C_A$ to a $Val_A > Val_B$ (and to a better neighboring solution at $NextA$), $Flag$ becomes $0$, and the values of $I_B$ are updated to $NextA$. When we reach a local optimum, $I_A$ and $I_B$ are stabilized to the same values. 

The workflow above is implemented by the \emph{Copy} and the \emph{Equality} gadgets. The \emph{CopyB} (resp. \emph{CopyA}) gadget updates the values of $I_A$ (resp. $I_B$) to the values in $NextB$ (resp. $NextA$), if $Val_A \leq Val_B$ and $Flag = 1$ (resp. $Val_A > Val_B$ and $Flag = 0$). When $Flag = 1$, the vertices in $T_B$ take the values of the vertices in $NextB$. If the values of $I_A$ and $NextB$ are different, the \emph{Equality} gadget sets the value of $Control_A$ to $0$. Hence, the \emph{Circuit Computing} gadget $C_A$ enters its write mode and the vertices in $I_A$ take the values of the vertices in $NextB$. Next, $Control_A$ becomes $1$, because the values of $I_A$ and $NextB$ are now identical, and $C_A$ enters its compute mode. As a result, the vertices in $ValA$ and $NextA$ take the values of $Real\text{-}Val(I_A)$ and $Real\text{-}Next(I_A)$, and we proceed to the next cycle. 

A key notion used throughout the reduction is the \emph{bias} that a vertex $i$ experiences wrt. a vertex subset. The bias of vertex $i$ wrt. (or from) $V' \subseteq V$ is $\big| \sum_{j \in V_i^1  \cap V'}w_j -\sum_{j \in V_i^0 \cap V' }w_j \big|$, where $V_i^1$ (resp. $V_i^0$) denotes the set of $i$'s neighbors on the $1$ (resp. $0$) side of the cut. 

\paragr{Technical Novelty.}
Next, we briefly discuss the points where our reduction needs to deviate substantially from known PLS reductions to \mc{}. Our discussion is unavoidlably technical and assumes familiarity with (at least one of) the reductions in \cite{ET11,GS10,SY91}.

Our \emph{Circuit Computing} gadgets are based on similar gadgets used e.g., in \cite{SY91}. A key difference is that our \emph{Circuit Computing} gadgets are designed to operate with $w_{Control}$ (i.e., weight of the Control vertex) arbitrarily smaller than the weight of any other vertex in the \emph{Circuit Computing} gadget. Hence, the Control vertex can achieve a very small bias wrt. the \emph{Circuit Computing} gadget (see Case~3, in Theorem~\ref{l:SY_gadgets}), which in turn, allows us to carefully balance the weights in the \emph{Equality} gadget. The latter weights should be large enough, so as to control the write and compute modes of $C_\ell$, and at the same time, small enough, so as not to interfere with the values of the input vertices $I_\ell$. The second important reason for setting $w_{Control}$ sufficiently small is that we need the Control vertex to follow the ``output'' of the \emph{Equality} gadget, and not the other way around. 

The discussion above highlights a key difference between \mc{} and \nmc{}. Previous reductions to \mc{} (see e.g., the reduction in \cite{SY91}) implement Control's functionality using different weights for its incident edges. More specifically, Control is connected with edges of appropriately large weight to the vertices of the circuit gadget, so that it can control the gadget's mode, and with edges of appropriately small weight to vertices outside the circuit gadget, so that its does not interfere with its other neighbors. 

For \nmc{}, we need to achieve the same desiderata with a single vertex weight. We manage to do so by introducing a \emph{Leverage} gadget (see Section~\ref{s:leverage} and cases~1~and~2, in Theorem~\ref{l:SY_gadgets}). Our \emph{Leverage} gadget reduces the influence of a vertex with large weight to a vertex with small weight and is used internally in the \emph{Circuit Computing} gadget. Hence, we achieve that Control has small bias wrt. the circuit gadget and weight comparable to the weights of the circuit gadget's internal vertices. 

Another important difference concerns the implementation of the \emph{red marks} (denoting the information flow) between \emph{Flag} and the \emph{Copy} gadgets, in Figure~\ref{f:construction}. They indicate that the value of \emph{Flag} should agree with the output of the \emph{Comparator} gadget. This part of the information flow is difficult to implement, because the \emph{Comparator} gadget and the \emph{Copy} gadgets receive input from $NextA$, $NextB$, $Val_A$ and $Val_B$, where the vertex weights are comparable to the weights of the output vertices in the \emph{Circuit Computing} gadgets. As a result, the weights of the vertices inside the \emph{Comparator} gadget cannot become sufficiently larger than the weights of the vertices inside the \emph{Copy} gadgets. \cite{SY91,GS10,ET11} connect \emph{Flag} to the \emph{Copy} gadgets with edges of sufficiently small weight, which makes the bias of \emph{Flag} from the \emph{Copy} gadgets negligible compared against its bias from \emph{Comparator}. Again, the \emph{Leverage} gadget comes to rescue. We use it internally in the \emph{Copy} gadgets, in order to decrease the influence of the vertices inside the \emph{Copy} gadgets to \emph{Flag}. As a result, \emph{Flag}'s bias from the \emph{Copy} gadgets becomes much smaller than its bias from \emph{Comparator} (see Lemma~\ref{l:comparator}).

Another key technical difference concerns the design of the \emph{Comparator} gadget. As stated in Lemma~\ref{l:comparator}, the \emph{Comparator} gadget computes the result of the comparison $Real\text{-}Val(I_A) \geq Real\text{-}Val(I_A)$, even if some ``input vertices'' have incorrect values. In previous work \cite{SY91,GS10,ET11}, \emph{Comparator} guarantees correctness of the values in both $NextB$ and $Val_B$ using appropriately chosen edge weights. With the correctness of the input values guaranteed, the comparison is not hard to implement. It is not clear if this decoupled architecture of the \emph{Comparator} gadget can be implemented in \nmc{}, due to the special structure of edge weights. Instead, we implement a new \emph{all at once} \emph{Comparator}, which ensures correctness to a subset of its input values enough to perform the comparison correctly.

%% file: conclusions.tex
\section{Conclusions and Future Work}
\label{sec:concl}

In this work, we showed that equilibrium computation in linear weighted congestion games is PLS-complete either on single-commodity series-parallel networks or on multi-commodity networks with identity latency functions, where computing an equibrium for (unweighted) congestion games is known to be easy. The key step for the latter reduction is to show that local optimum computation for \nmc{}, a natural and significant restriction of \mc{}, is PLS-complete. The reductions in Section~\ref{sec:congestion} are both \emph{tight} \cite{SY91}, thus preserving the structure of the local search graph. In particular, for the first reduction, we have that (i)~there are instances of linear weighted congestion games on single-commodity series-parallel networks such that any best response sequence has exponential length; and (ii)~that the problem of computing the equilibrium reached from a given initial state is PSPACE-hard.

However, our reduction of \cf{} to \lnmc{} is not tight. Specifically, our \emph{Copy} and \emph{Equality} gadgets allow that the \emph{Circuit Computing} gadget might enter its \emph{compute} mode, before the entire input has changed. Thus, we might ``jump ahead'' and reach an equilibrium before \cf{} would allow, preventing the reduction from being tight.

Our work leaves leaves several interesting directions for further research. A natural first step is to investigate the complexity of equilibrium computation for weighted congestion games on series-parallel (or extension-parallel) networks with identity latency functions. An intriguing research direction is to investigate whether our ideas (and gadgets) in the PLS-reduction for \nmc{} could lead to PLS-hardness results for approximate equilibrium computation for standard and weighted congestion games (similarly to the results of Skopalik and V\"ocking \cite{SV08}, but for latency functions with non-negative coefficients). Finally, it would be interesting to understand better the quality of efficiently computable approximate equilibria for \nmc{} and the smoothed complexity of its local optima.

%% file: sepaPLS.tex
\subsection{The Proof of Theorem~\ref{thm:sepaPLShard}}\label{sec:singlePLS}

%
%
We will reduce from the PLS-complete problem \lmc{} and given an instance of \mc{} we will construct a network weighted network  Congestion Game for which the Nash equilibria will correspond to maximal solutions of \lmc{} and vice versa. First we give the construction and then we  prove the theorem. For the  formal PLS-reduction, which needs functions $\phi_1$ and $\phi_2$, $\phi_1$ returns the (polynomially) constructed instance described below and $\phi_2$ will be revealed later in the proof.

Let $H(V,E)$ be an edge-weighted graph of a \lmc{} instance and let $n=|V|$ and $m=|E|$. In the constructed network weighted CG  instance there will be $3n$ players which will share $n$ different weights inside the set $\{16^i:i\in [n]\}$ so that for every $i\in[n]$ there are exactly 3 players having weight $w_i=16^i$.  All  players share a common origin-destination pair $o-d$ and choose $o-d$  paths on a series-parallel graph $G$.
Graph $G$ is a parallel composition  of two identical  copies of a series-parallel graph. Call these copies $G_1$ and $G_2$. In turn, each of $G_1$ and $G_2$ is a series composition of $m$ different series-parallel graphs, each of which corresponds to the $m$ edges of $H$. For every $\{i,j\}\in E$ let $F_{ij}$ be the series-parallel graph that corresponds to $\{i,j\}$. Next we describe the construction of   $F_{ij}$, also shown in Fig. \ref{fig:sepaPLS}.


$F_{ij}$ has 3 vertices, namely $o_{ij},v_{ij}$ and $d_{ij}$ and $n+1$ edges. For any $k\in [n]$ other than $i,j$ there is an $o_{ij}-d_{ij}$ edge with latency function $\ell_k(x)=\frac{Dx}{4^k}$, where $D$ serves as a big constant to be defined later.  There are also two $o_{ij}-v_{ij}$ edges, one with latency function  $\ell_i(x)=\frac{Dx}{4^i}$ and one with latency function  $\ell_j(x)=\frac{Dx}{4^j}$. Last, there is a $v_{ij}-d_{ij}$ edge with latency function  $\ell_{ij}(x)=\frac{w_{ij}x}{w_iw_j}$, where $w_{ij}$ is the weight of edge $\{i,j\}\in E$ and $w_i$ and $w_j$ are the weights of  players $i$ and $j$, respectively, as described earlier. Note that in every $F_{ij}$ and for any  $k\in [n]$ the latency function $\ell_k(x)=\frac{Dx}{k}$  appears in exactly one edge. With $F_{ij}$ defined, an example of the structure of such a network $G$ is given in  Fig. \ref{fig:sepaExample}. 

Observe that in each of $G_1$ and $G_2$ there is a unique  path that contains all the edges with latency functions $\ell_i(x)$, for $i\in[n]$, and call these paths $p^u_i$ and $p_i^l$ for the upper ($G_1$) and lower ($G_2$) copy respectively. Note that each of $p^u_i$ and $p_i^l$ in addition to those edges, contains some edges with latency function of the form $\frac{w_{ij}x}{w_iw_j}$. These edges for path $p^u_i$ or $p_i^l$ is in one to one correspondence to the edges of vertex $i$ in $H$  and this is crucial for the proof.


We go on to prove the correspondence of Nash equilibria in $G$ to maximal cuts in $H$, i.e., solutions of \lmc. We will first show that at a Nash equilibrium, a player of weight $w_i$ chooses either $p^u_i$ or $p_i^l$. Additionally, we prove that $p^u_i$ and $p_i^l$  will have  at least one player (of weight $w_i$). 
This already provides a good structure of a Nash equilibrium and players of different weights, say $w_i$ and $w_j$, may go through the same edge in $G$ (the edge with latency function $w_{ij}x/w_iw_j$)  only if  $\{i,j\}\in E$. The correctness of the reduction lies in the fact that players in $G$ try to minimize their costs incurred by these type of edges  in the same way  one wants to  minimize the sum of the weights of the edges in each side of the cut when solving \lmc. 

To begin with, we will prove that  at equilibrium any player of weight $w_i$ chooses either $p^u_i$ or $p_i^l$ and at least one such player chooses each of $p^u_i$ and $p_i^l$. For that, we will need the following proposition as a building block, which will also reveal a suitable value for $D$.

\begin{proposition}\label{prop:PlayersToPaths}
For some $i,j\in[n]$  consider $F_{ij}$ (Fig. \ref{fig:sepaPLS}) and assume that for all $k\in[n]$,  there are either one, two or three players of weight $w_k$ that have to choose an $o_{ij}-d_{ij}$ path. At equilibrium,  all players of weight $w_k$ (for any $k\in[n]$) will go through the path 
that contains a edge  with latency function $\ell_k(x)$. 
\end{proposition}

\begin{proof}
The proof is by induction on the different weights starting from bigger weights. For any $k\in [n]$ call $e_k$ the edge of $F_{ij}$ with latency function $\ell_k(x)$ and call $e_{ij}$ the edge with latency function $\frac{w_{ij}x}{w_iw_j}$. For some $k\in[n]$ assume that for all $l>k$ all players of weight $w_l$ have chosen the path containing $e_l$ and lets prove that this is the case for players of weight $w_k$ as well. Since $D$ is going to be big enough, for the moment ignore edge $e_{ij}$ and assume that in $F_{ij}$ there are only $n$ parallel paths each consisting of a single edge. 

Let the players be at equilibrium and consider any player, say player $K$, of weight $w_k$. The cost she computes on $e_k$ is upper bounded by the cost of $e_k$  if all players with weight up to $w_k$ are on $e_k$, since by induction players with weight $>w_k$ are not on $e_k$ at equilibrium. This cost is upper bounded by  $c^k=\frac{D(3\sum_{l=1}^{k}16^l)}{4^k}=\frac{3D\frac{16^{k+1}-1}{16-1}}{4^k}$.


For any edge $e_l$ for $l<k$, the cost that $K$ computes is  lower bounded by $c^<=\frac{D16^k}{4^{k-1}}$ 
since she must include herself in the load of $e_l$ and the edge with the smallest slope in its latency function is $e_{k-1}$. But then $c^k<c^<$, since 
\[
c^k  < c^<  \Leftrightarrow
\frac{3D\frac{16^{k+1}-1}{16-1}}{4^k} <
\frac{D16^k}{4^{k-1}}\Leftrightarrow
48\cdot16^{k}-3 < 60\cdot16^k
\]

Thus, at equilibrium  players of weight $w_k$ cannot be on any of the $e_l$'s for all $l<k$. On the other hand, the cost that $K$ computes for $e_l$ for $l>k$ is at least $c_l^>=\frac{D(16^l+16^k)}{4^l}$, since by induction $e_l$ is already chosen by at least one player of weight $w_l$. But then $c^k<c^>_l$ since
\begin{align*}
\frac{3D\frac{16^{k+1}-1}{16-1}}{4^k} &<\frac{D(16^l+16^k)}{4^l}\Leftrightarrow \\
48\cdot16^{k}-3 & <15\frac{16^{l}+16^k}{4^{l-k}}\Leftrightarrow \\
48\cdot4^{l-k}16^k &< 15\cdot16^l=15\cdot4^{l-k}4^{l-k}16^k \,.
\end{align*}
Thus, at equilibrium  players of weight $w_k$ cannot be on any of the $e_l$'s for all $l>k$. 

This completes the induction for the simplified case where we ignored the existence of $e_{ij}$, but lets go on to include it and define $D$ so that the same analysis goes through.
By the above, $c^<-c^k=\frac{D16^k}{4^{k-1}}-\frac{3D\frac{16^{k+1}-1}{16-1}}{4^k}>D$ and also for any $l>k$ it is $c^>_l-c^k=\frac{D(16^l+16^k)}{4^l}-\frac{3D\frac{16^{k+1}-1}{16-1}}{4^k}>D$ (this difference is minimized for $l=k+1$). On the other hand the maximum cost that edge $e_{ij}$ may have is bounded above by $c^{ij}=\frac{w_{ij}3\sum_{l=1}^n16^l}{w_iw_j}$, as $e_{ij}$ can be chosen by at most all of the players and note that $c^{ij}\leq 16^{n+1}\max_{q,r\in [n]}w_{qr}$. Thus, one can choose a big value for $D$, namely $D=16^{n+1}\max_{q,r\in [n]}w_{qr}$, so that even if  a player with weight $w_k$ has to add the cost of $e_{ij}$  when computing her path cost, it still is $c^{ij}+c^k<c^<$ (since $c^<-c^k>D\geq c_{ij}$) and for all $l>k$: $c^{ij}+c^k<c^>_l$ (since $c^>_l-c^k>D\geq c^{ij}$), implying that at equilibrium all players of weight $w_k$ may only choose the path through $e_k$. 
\end{proof}

Other than revealing a value for $D$, the proof of Porposition \ref{prop:PlayersToPaths}  reveals a crucial property: a player of weight $w_k$ in $F_{ij}$ strictly prefers the path containing $e_k$ to the path containing $e_l$ for any $l<k$, independent to whether players of weight $>w_k$ are present in the game or not. With this in mind we go back to prove that at equilibrium any player of weight $w_i$ chooses either $p^u_i$ or $p_i^l$ and at least one such player chooses each of $p^u_i$ and $p_i^l$. The proof is by induction, starting from  bigger weights.

Assume that by the inductive hypothesis for every $i>k$, players with weights  $w_i$ have chosen paths $p^u_i$ or $p_i^l$ and at least one such player chooses each of $p^u_i$ and $p_i^l$. Consider a player of weight $w_k$, and,  wlog, let her have chosen an $o-d$ path through $G_1$. Since at least one player for every bigger weight is by induction already in the paths of $G_1$ (each in her corresponding $p_i^u$), Proposition \ref{prop:PlayersToPaths} and the remark after its proof give that in each of the $F_{ij}$'s the player of weight $w_k$ has chosen the subpath of $p_k^u$, and this may happen only if her chosen path is $p_k^u$. It remains to show that there is another player of weight $w_k$ that goes through $G_2$, which, with an argument similar to the previous one, is equivalent to this player choosing path $p_k^l$.

To reach a contradiction, let  $p_k^u$ be chosen by all three players of weight $w_k$, which leaves $p_k^l$ empty. Since all players of bigger weights are by induction  settled in paths completely disjoint to $p_k^l$, the load on this path if we include a player of weight $w_k$ is upper bounded by  the sum of  all players of weight $<w_k$ plus $w_k$, i.e., $16^k+3\sum_{t=1}^{k-1}16^t=16^k+3\frac{16^k-1}{16-1}$, which is less than the lower bound on the load of  $p_k^u$, i.e., $3\cdot16^k$ (since $p_k^u$ carries 3 players of weight $16^k$). This already is a contradiction to the equilibrium property, since $p_k^u$ and  $p_k^l$ share the exact same latency functions on their edges which, given the above inequality on the loads, makes $p_k^u$ more costly than $p_k^l$ for a player of weight $w_k$. To summarize, we have the following.

\begin{proposition}\label{prop:playersToPaths2}
At equilibrium, for every $i\in[n]$ a player of weight $w_i$ chooses either $p_i^u$ or $p_i^l$. Additionally, each of $p_i^u$ and  $p_i^l$ have been chosen by at least one player (of weight $w_i$).
\end{proposition}

Finally, we prove that every equilibrium of the constructed instance corresponds to a maximal solution of \lmc{} and vice versa. Given a maximal solution $S$ of \lmc{} we will show that  the configuration $Q$ that for every $k\in S$ routes $2$ players through $p_k^u$ and $1$ player through $p_k^l$ and for every $k\in V\setminus S$ routes $1$ player through $p_k^u$ and $2$ players through $p_k^l$ is an equilibrium. Conversely, given an equilibrium $Q$ the cut $S=\{k\in V:\mbox{2 players have chosen $p_k^u$ at $Q$}\}$  is a maximal solution of \lmc. 

Assume that we are at equilibrium and consider a player of weight $w_k$ that has chosen $p_k^u$ and wlog $p_k^u$ is chosen by two players (of weight $w_k$). By the equilibrium conditions the cost she computes for $p_k^u$ is at most the cost she computes for $p_k^l$,  which, given Proposition \ref{prop:playersToPaths2}, implies
$$\sum_{i=1}^m\frac{2D16^k}{4^k}+\sum_{\{k,j\}\in E}\frac{w_{kj}(2\cdot16^k+x_j^u16^j)}{16^k16^j} \leq \sum_{i=1}^m\frac{2D16^k}{4^k}+\sum_{\{k,j\}\in E}\frac{w_{kj}(2\cdot16^k+x_j^l16^j)}{16^k16^j}$$
where $x^u_j$ (resp. $x_j^l$) is either $1$ or $2$ (resp. $2$ or $1$) depending whether, for any $j:\{k,j\}\in E$, one or two players (of weight $w_j$) respectively have chosen path $p_j^u$. By canceling out terms, the above implies 
\begin{equation}\label{eqn:LmcToEqCorrespondence}
\sum_{\{k,j\}\in E}w_{kj}x_j^u \leq \sum_{\{k,j\}\in E}w_{kj}x_j^l\Leftrightarrow \sum_{\{k,j\}\in E}w_{kj}(x_j^u-1) \leq \sum_{\{k,j\}\in E}w_{kj}(x_j^l-1)
\end{equation}

Define $S=\{i\in V:x_i^u=2\}$. By our assumption it is  $k\in S$ and the left side  of (\ref{eqn:LmcToEqCorrespondence}), i.e., $\sum_{\{k,j\}\in E}w_{kj}(x_j^u-1)$, is the sum of the weights of the edges of $H$ with one of its vertices being $k$ and the other belonging in $S$. Similarly, the right side of of (\ref{eqn:LmcToEqCorrespondence}), i.e., $\sum_{\{k,j\}\in E}w_{kj}(x_j^l-1)$ is the  sum of the weights of the edges with one of its vertices being $k$ and the other belonging in $V\setminus S$. But then  (\ref{eqn:LmcToEqCorrespondence}) directly implies  that for the (neighboring) cut $S'$ where  $k$  goes from $S$ to $V\setminus S$ it holds $w(S)\geq w(S')$. Since $k$ was arbitrary (given the symmetry of the problem), this holds for every $k\in [n]$ and thus for every $S'\in N(S)$ it is  $w(S)\geq w(S')$ proving one direction of the claim. Observing that the argument works backwards we complete the proof.   For the formal proof, to define function $\phi_2$, given the constructed instance and one of its solutions, say $s'$, $\phi_2$ returns  solution $s=\{k\in V:\mbox{2 players have chosen $p_k^u$ at $s'$}\}$. \qed 

%% file: multiComPLS.tex
\subsection{The Proof of Theorem~\ref{thm:mutiPLS}}
\label{sec:multiPLS}

We will reduce from the PLS-complete problem \lnmc{}. Our construction draws ideas from Ackermann et al. \cite{ARV08}. For an instance of \lnmc{} we will construct a multi-commodity network CG where every equilibrium will correspond to a maximal solution of \lnmc{}  and vice versa. For the  formal PLS-reduction, which needs functions $\phi_1$ and $\phi_2$, $\phi_1$ returns the (polynomially) constructed instance described below and $\phi_2$ will be revealed later in the proof.

We will use only the identity function as the  latency function of every edge, but for ease of presentation
we will first prove our claim assuming we can use constant latency functions on the edges. Then we will describe how we can drop this assumption and  use only the identity function on all edges, and have the proof still going through.

Let $H(V,E)$ be the vertex-weighted graph of an instance of \lnmc{} with $n=|V|$ vertices and $m=|E|$ edges. The network weighted congestion game has $n$ players, with player $i$ having her own origin destination $o_i-d_i$ pair and weight $w_i$ equal to the weight of vertex $i\in V$. In the constructed network there will be many $o_i-d_i$ paths for every player $i$ but there will be exactly two paths that cost-wise dominate all others. At equilibrium, every player will choose one of these two paths that correspond to her. This choice for player $i$ will be equivalent to picking the side of the cut that vertex $i$ should lie in order to get a maximal solution of \lnmc{}.

\begin{figure}
    \centering
    \includegraphics[scale=0.5]{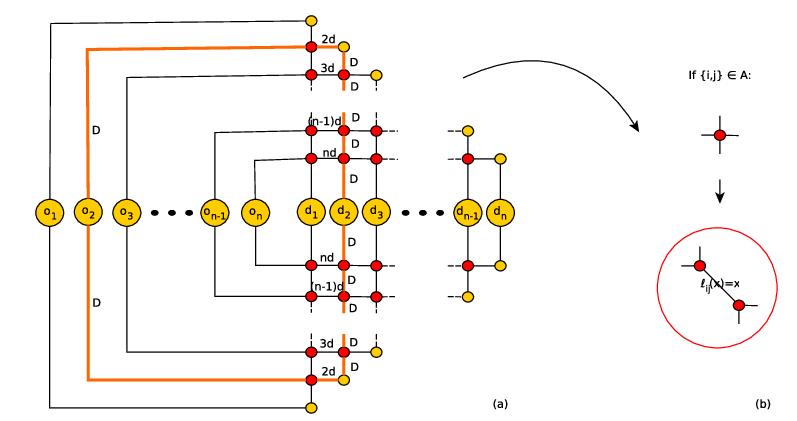}
    \caption{(a) The construction of the reduction of Theorem \ref{thm:mutiPLS}. As an example, in orange are  the least costly $o_2-d_2$ paths $p_2^u$ (up) and $p_2^l$ (down), each with cost equal to $2D+2d+(n-2)D$. (b) The replacement of the red vertex at the $i$-th row and $j$-th column of the upper half-grid whenever edge $\{i,j\}\in E$. A symmetric replacement happens in the lower half-grid.}
    \label{fig:multiPLS}
\end{figure}

The initial network construction is shown in Fig. \ref{fig:multiPLS}. It has $n$ origins and $n$ destinations. The rest of the vertices lie either on  the lower-left half (including the diagonal) of a  $n\times n$ grid, which we call the upper part, or the upper-left half   of another $n\times n$ grid, which we call the lower part. Other than the edges of the two half-grids that are all present,  there are edges connecting the origins and the destinations to the two parts. For $i\in [n]$, origin $o_i$ in each of the upper and lower parts connects  to the first (from left to right) vertex  of the row that has $i$ vertices in total. For $i\in [n]$, destination $d_i$ in each of the upper and lower parts connects to the   $i$-th vertex of the row that has $n$ vertices in total. To define the (constant) latency functions, we will need 2 big constants, say $d$ and  $D=n^3d$, 
and note that $D\gg d$. 

All edges that connect to an origin or a destination 
and all the vertical edges of the half-grids will have  constant $D$ as their latency function, and any horizontal edge that lies on a row with $i$ vertices  will have  constant $i\cdot d$ as its latency function. 
To finalize the construction we will do some small changes but note that, as it is now, player $i$ has two shortest paths that are far less costly (at least by $d$) than all other paths. These two paths are path $p^u_i$ that starts at $o_i$, continues horizontally through the upper part for as much as it can and then continues vertically to reach $d_i$, and path $p^l_i$ which does the exact same thing through the lower part (for an example see Fig. \ref{fig:multiPLS}a). Each of $p^u_i$ and $p^l_i$ costs equal to  $c_i=2D+i(i-1)d+(n-i)D$. To verify this claim simply note that (i) if a path  tries to go through another origin or moves vertically away from $d_i$ in order to reach less costly horizontal edges, then it will have to pass through at least $(2+i-1)+2$ vertical edges of cost $D$ and  its cost from such edges compared to $p^u_i$'s and  $p^l_i$'s costs increases  by at least $2D=2n^3d$,  which is already more than paying all horizontal edges;  and (ii) if it moves vertically towards $d_i$ earlier than $p^u_i$ or $p^l_i$ then its cost increases by at least $d$, since it moves towards more costly horizontal edges.

To complete the construction if $\{i,j\}\in E$ (with wlog $i<j$) we replace the (red) vertex at position $i,j$ of the upper and the lower half-grid ($1,1$ is top left for the upper half-grid and lower left for the lower half-grid) with two vertices connected with an edge, say $e^u_{ij}$ and $e_{ij}^l$ respectively, with latency function $\ell_{ij}(x)=x$, where the first vertex connects with the vertices at positions $i,j-1$ and $i-1,j$ of the grid and the second vertex connects to the vertices at positions $i+1,j$ and $i,j+1$ (see also Fig. \ref{fig:multiPLS}b). Note that if we take $d\gg \sum_{k\in [n]}w_k$, then,  for any $i\in [n]$, paths $p^u_i$ and $p^l_i$  still have significantly lower  costs  than all other $o_i-d_i$ paths. Additionally, if $\{i,j\}\in E$ then $p^u_i$ and $p^u_j$ have a single common edge and $p^l_i$ and $p^l_j$ have a single common edge, namely $e^u_{ij}$ and $e_{ij}^l$ respectively,  which add some extra cost to the paths (added to $c_i$ defined above).

Assume we are at equilibrium. By the above discussion player $i\in[n]$ may only have chosen $p^u_i$ or $p^l_i$. Let $S=\{i\in [n]: \mbox{player $i$ has chosen $p^u_i$}\}$. We will prove that $S$ is a solution to \lnmc{}. By the equilibrium conditions for every $i\in S$ the cost of $p_i^u$, say $c_i^u$, is less than or equal to the cost of $p^l_i$, say $c_i^l$. Given the choices of the rest of the players, and by defining $S_i$ to be the neighbors of $i$ in $S$, i.e. $S_i=\{j\in S:\{i,j\}\in E\}$, and $V_i$ be the neighbors of  $i$ in $V$, i.e. $V_i=\{j\in V:\{i,j\}\in E\}$,  $c_i^u\leq c_i^l$ translates to 
\begin{multline*}
\Big(2D+i(i-1)d+(n-i)D\Big)+\Big(\sum_{j\in V_i}w_i+\sum_{j\in S_i}w_j\Big)\\
\leq
\Big(2D+i(i-1)d+(n-i)D\Big)+\Big(\sum_{j\in V_i}w_i+\sum_{j\in V_i\setminus S_i}w_j\Big)\,,
\end{multline*}
with the costs in the second and fourth parenthesis coming from the $e_{ij}$'s for the different $j$'s. This equivalently gives $$\sum_{j\in S_i}w_j\leq\sum_{j\in V_i\setminus S_i}w_j\Leftrightarrow \sum_{j\in S_i}w_iw_j\leq\sum_{j\in V_i\setminus S_i}w_iw_j.$$ 

The right side of the last inequality equals to the weight of the edges with $i$ as an endpoint that cross  cut $S$. The left side equals to the weight of the edges with $i$ as an endpoint that cross the cut $S'$, where $S'$ is obtained by moving $i$ from $S$ to $V\setminus S$. Thus for $S$ and $S'$ it is $w(S')\leq w(S)$. A similar argument (or just symmetry) shows that if $i\in V\setminus S$ and we send $i$ from $V\setminus S$ to  $S$ to form a cut $S'$ it would again be  $w(S')\leq w(S)$. Thus, for any $S'\in N(S)$ it is $w(S)\geq w(S')$ showing that  $S$ is a solution to \lnmc{}. Observing that the argument works backwards we have that from an arbitrary solution of \lnmc{} we may get an equilibrium for the constructed weighted CG  instance. For the formal part, to define function $\phi_2$, given the constructed instance and one of its solutions, $\phi_2$ returns  solution $s=\{i\in [n]: \mbox{player $i$ has chosen $p^u_i$}\}$.

What remains to show is how we can almost simulate the constant latency functions so that we use only the identity function on all edges and, for every $i\in [n]$, player $i$ still may only choose paths $p^u_i$ or $p^l_i$ at equilibrium.  Observe that,  
since we have a multi-commodity instance we can simulate (exponentially large) constants by replacing an edge $\{j,k\}$ with a three edge path $j-o_{jk}-d_{jk}-k$, adding a complementary player with origin $o_{jk}$ and destination $d_{jk}$ and weight equal to the desired constant. Depending on the rest of the structure we may additionally  have to make sure (by suitably defining latency functions) that this player prefers going through  edge $\{o_{jk},d_{jk}\}$ at equilibrium. 

To begin with, consider any horizontal edge $\{j,k\}$ with latency function $i\cdot d$ (for some $i\in [n]$) and replace it with a three edge path $j-o_{jk}-d_{jk}-k$. Add a player  with origin $o_{jk}$ and destination $d_{jk}$ with weight equal to $i n^3 w$, where $w=\sum_{i\in [n]}w_i$,  and let all edges have the identity function. At equilibrium no matter the sum of the weights of the players that choose this three edge path, the $o_{jk}-d_{jk}$ player prefers to use the direct $o_{jk}-d_{jk}$ edge or else she pays at least double the cost (middle edge vs first and third edges). Thus the above replacement is (at equilibrium) equivalent to having edge $\{j,k\}$ with latency function $3x+in^3w=3x+i\cdot d$, for $d=n^3w$.

 Similarly, consider any  edge $\{j,k\}$ with latency function $D$  and replace it with a three edge path $j-o_{jk}-d_{jk}-k$. Add a player  with origin $o_{jk}$ and destination $d_{jk}$ with weight equal to $n^3d$  and let all edges have the identity function. Similar to above, 
 this replacement is (at equilibrium) equivalent to having edge $\{j,k\}$ with latency function $3x+n^3d=3x+D$, for $D=n^3d$. 

With these definitions, at equilibrium, all complementary players will go through the correct edges and, due to the complementary players, all edges that connect to an origin or a destination will have cost $\approx D$, all vertical edges of the half-grids will cost  $\approx D$, and any horizontal edge that lies on a row with $i$ vertices will cost  $\approx i\cdot d$, where ``$\approx$'' means at most within $\pm 3w=\pm\frac{3d}{n^3}$ (note that $w$ is the maximum weight that the $o_i-d_i$ players can add to each of the three edge paths). Additionally, for every $i\in [n]$,  $p^u_i$ and  $p^l_i$ are structurally identical, i.e., they have the same structure, identical complementary players on their edges and share the same latency functions. All the above make the analysis go through in the same way as in the simplified construction. \qed

%% file: Appendix_Algo.tex
\section{Missing Technical Details from the Analysis of~\algo{}}
\label{app:algo}

In this section, we present an algorithm that computes approximate equilibria for \nmc{}. Let $G(V,E)$ be vertex-weighted graph with $n$ vertices and $m$ edges, and consider any $\eps>0$. The algorithm, called \algo{} and formally presented in Algorithm \ref{Algorithm}, returns a $(1+\eps)^3$-approximate equilibrium (Lemma~\ref{lem:apxGrntAlgo}) for $G$ in time $O(\frac{m}{\eps}\lceil\frac{n}{\eps}\rceil^{2D_\eps})$ (Lemma~\ref{lem:runTimeAlgo}), where $D_\eps$ is the number of different rounded weights, i.e., the  weights produced by rounding down each of the original weights to its closest power of $(1+\eps)$. To get a $(1+\eps$)-approximate equilibrium, for $\eps < 1$,  it suffices to run the algorithm with $\eps'=\frac{\eps}{7}$.

\paragr{Description of the Algorithm.}
\algo{}  first creates an instance $G'$ with weights rounded down  to their closest power of $(1+\eps)$, i.e., weight $w_i$ is replaced by weight $w'_i=(1+\eps)^{\lfloor\log_{1+\eps}w_i\rfloor}$ in  $G'$, and then computes a  $(1+\eps)^2$-approximate equilibrium for $G'$. Observe that  any $(1+\eps)^2$-approximate equilibrium for $G'$ is a $(1+\eps)^3$-approximate equilibrium for $G$, since 
\[\sum_{j \in V_i : s_i = s_j}w_j\leq(1+\eps)\sum_{j \in V_i: s_i = s_j}w'_j \leq (1+\eps)^3 \sum_{j \in V_i : s_i \neq s_j}w'_j\leq (1+\eps)^3 \sum_{j \in V_i: s_i \neq s_j}w_j,\]
where $V_i$ denotes the set of vertices that share an edge with vertex $i$, with the first and  third inequalities following from the rounding and the second one following from the  equilibrium condition for $G'$. 

To compute a $(1+\eps)^2$-approximate equilibrium, \algo{} first sorts the vertices in increasing weight order and note that, wlog, we may assume that $w'_{1}=1$, as we may simply divide all weights by $w'_1$. Then, it groups the vertices so that the fraction of the weights of consecutive vertices in the same group is bounded above by $\lceil n/\eps\rceil$, i.e., for any $i$,  vertices $i$ and $i+1$ belong in the same group if and only if $\frac{w'_{i+1}}{w'_{i}}\leq\lceil\frac{n}{\eps}\rceil$. This way, groups $g_j$ are formed on which we assume an increasing order, i.e., for any $j$, the vertices in $g_j$ have smaller weights than those in $g_{j+1}$.

The next step is to bring the groups closer together using the following process which will generate weights $w''_i$. For all $j$,  all the weights of vertices on heavier groups, i.e., groups $g_{j+1},g_{j+2},\ldots$, are divided by  $d_j=\frac{1}{\lceil n/\eps\rceil}\frac{w_{j+1}^{min}}{w_j^{max}}$ so that $\frac{w_{j+1}^{min}/d_j}{w_j^{max}}=\lceil\frac{n}{\eps}\rceil$, where $w_{j+1}^{min}$ is the smallest weight  in $g_{j+1}$ and $w_j^{max}$ is the biggest weight in $g_{j}$. For vertex $i$, let the resulting weight be $w''_i$, i.e., $w''_i=\frac{w'_i}{\Pi_{j\in I_i}d_j}$, where $I_i$ contains the indexes of groups below i's group, and keep the increasing order on the vertex weights. Observe that by the above process for any $i$: $\frac{w''_{i+1}}{w''_{i}}\leq \lceil\frac{n}{\eps}\rceil$, either because $i$ and $i+1$ are in the same group or because the groups are brought closer together. Additionally if $i$ and $i+1$ belong in different groups then $\frac{w''_{i+1}}{w''_{i}}= \lceil\frac{n}{\eps}\rceil$, implying that for vertices $i,i'$ in different groups with $w''_{i'}>w''_i$ it is $\frac{w''_{i'}}{w''_{i}}\geq \lceil\frac{n}{\eps}\rceil$.  Thus, if we let $D_\eps$ be the number of different weights in $G'$, i.e., $D_{\eps}=|\{w'_i:i\,\, vertex \,\,of\,\,G'\}|$, then the maximum weight $w''_n$ 
is $w''_{n}
=\frac{w''_{n}}{w''_{n-1}}\frac{w''_{n-1}}{w''_{n-2}}\ldots\frac{w''_{2}}{w''_{1}}\leq \lceil\frac{n}{\eps}\rceil^{D_\eps}$


In a last step, using the $w''$ weights, \algo{} starts from an arbitrary configuration (a 0-1 vector) and lets the vertices play $\eps$-best response moves, i.e., as long as there is an index $i$ of the vector violating the $(1+\eps)$-approximate equilibrium condition, \algo{} flips its bit. When there is no such index  \algo{}  ends and returns the resulting configuration.

\begin{figure}[ht]
\begin{algorithm}[H]
\caption{\algo{}, computing  $(1+\eps)^3$-approximate equilibria\label{Algorithm}}

\KwIn{A \nmc{} instance $G(V, E)$ with $n$ vertices and weights  $\{w_i\}_{i\in [n]}$ sorted increasingly with $w_1=1$, and an $\eps>0$.}
\KwOut{A vector $\vec{s}\in\{0,1\}^n$ partitioning  the vertices in two sets.}

~\\
\nl \textbf{for}  $i\in[n]$ \textbf{do} $w_i:=(1+\eps)^{\lfloor\log_{1+\eps}w_i\rfloor}$

\nl groups:= 1;

\textbf{insert} $w_1$ \textbf{into} $g_{groups}$; \Comment{Assign the weights into groups $\{g_j\}_{j\in [groups]}$}

\For{$i\in \{2,...,n\}$}{

\textbf{if} $\frac{w_i}{w_{i-1}}> \lceil \frac{n}{\epsilon}\rceil$ \textbf{then} groups++;

\textbf{insert} $w_i$ \textbf{into} $g_{groups}$;
}

\nl \For {$j\in \{2,...,groups\}$}{
$w_{j}^{min}:=$ minimum weight of group $g_j$;  \Comment{Bring the groups $\lceil \frac{n}{\eps}\rceil$ close}

$w_{j-1}^{max}:=$ maximum weight of group $g_{j-1}$;

$d_j=\frac{1}{\lceil n/\eps\rceil}\frac{w_{j+1}^{min}}{w_j^{max}}$ 

\textbf{for} $w_i\in g_j\cup\ldots\cup g_{group}$ \textbf{do} $w_i:=w_i/d_j$;
}

\nl $\vec{s}$:= an arbitrary $\{0,1\}^n$ vector;

\nl For all $i$, let $V_i = \{ j : \{ i, j \} \in E \}$ be the neighborhood of $i$ in $G$;

\vspace{1mm}
\While{$\exists i:\sum_{j \in  V_i: s_i = s_j}w_j > (1+\eps)\sum_{j \in  V_i: s_i \neq s_j}w_j$}{\vspace{1mm}$s_i:=1-s_i$;\Comment{Moves towards equilibrium}}

\nl \Return $\vec{s}$.
\end{algorithm}
\end{figure}

\begin{lemma}\label{lem:runTimeAlgo}
For any $\eps > 0$, \algo{} terminates in time $O(\frac{m}{\eps}\lceil\frac{n}{\eps}\rceil^{2D_\eps})$
\end{lemma}

\begin{proof}
We are going to show the claimed bound for the last step of \algo{} since all previous steps  can be (naively) implemented to end in $O(n^2)$ time.

The proof relies on a potential function argument.
For any $\vec{s}\in\{0,1\}^n$, let $$\Phi(\vec{s})=\frac{1}{2}\sum_{i\in V}\sum_{j\in V_i :s_i = s_j}w''_i
w''_j .$$
%
Since for the maximum weight $w''_{n}$ it is $w_{n}''\leq\lceil\frac{n}{\eps}\rceil^{D_\eps}$, it follows that $\Phi(\vec{s})\leq m\lceil\frac{n}{\eps}\rceil^{2D_\eps}$. On the other hand whenever an $\eps$-best response move is made by \algo{}  producing $\vec{s}\hspace{0.5mm}'$ from some $\vec{s}$,  $\Phi$ decreases by at least $\eps$, i.e., $\Phi(\vec{s})-\Phi(\vec{s}\hspace{0.5mm}')\geq\eps$. This is because, if $i$ is the index flipping bit from $\vec{s}$ to $\vec{s}\hspace{0.5mm}'$, then  by the violation of the $(1+\eps)$-equilibrium condition \[w''_i\sum_{j \in V_i: s_i = s_j}w''_j \geq w''_i(1+\eps)\sum_{j \in V_i: s_i \neq s_j}w''_j\Rightarrow \sum_{j \in V_i: s_i = s_j}w''_iw''_j-\sum_{j \in V_i: s'_i = s'_j}w''_iw''_j \geq \eps,\]
since $\sum_{j \in V_i: s'_i = s'_j}w''_iw''_j\geq1$, and 
\[\Phi(\vec{s})-\Phi(\vec{s}\hspace{0.5mm}')=\sum_{j \in V_i: s_i = s_j}w''_iw''_j-\sum_{j \in V_i: s'_i = s'_j}w''_i
w''_j\geq\eps.\]
Consequently, the last step of the algorithm will do at most $\frac{m\lceil\frac{n}{\eps}\rceil^{2D_\eps}}{\eps}$ $\eps$-best response moves. 
\end{proof}

\begin{lemma}\label{lem:apxGrntAlgo}
For any vertex-weight graph $G$ and any $\eps>0$, \algo{} returns a $(1+\eps)^3$-approximate equilibrium for \nmc{} in $G$.
\end{lemma}

\begin{proof}
Clearly, \algo{} terminates with a vector $\vec{s}$ that is a  $(1+\eps)$-approximate equilibrium for the instance with the $w''$ weights. It suffices to show that $\vec{s}$ is a $(1+\eps)^2$-approximate equilibrium for $G'$, i.e., the instance with the $w'$ weights, since this will directly imply that $\vec{s}$ is a $(1+\eps)^3$-approximate equilibrium for $G$, as already discussed at the beginning of the description of the algorithm.

Consider any index $i$ and let $V_i^{h}$ be the neighbors of $i$ that belong in the heaviest group among the neighbors of $i$. By the $(1+\eps)$-approximate equilibrium condition it is  
\begin{equation}\label{eq:alg1}
\sum_{j \in V_i\setminus V_i^{h}: s_i = s_j}w''_j+\sum_{j \in  V_i^{h}: s_i = s_j}w''_j \leq (1+\eps)\big(\sum_{j \in V_i\setminus V_i^{h}: s_i \neq s_j}w''_j+\sum_{j \in  V_i^{h}: s_i \neq s_j}w''_j\big).
\end{equation}
Recalling that for every $j$, $w''_j=\frac{w'_j}{\Pi_{k\in I_j}d_k}$, where $I_j$ contains the indexes of groups below j's group, and letting $D=\Pi_{k\in I_j}d_k$,  for a $j\in V_i^{h}$, gives
 \begin{equation}\label{eq:alg2}
 \sum_{j \in V_i\setminus V_i^{h}: s_i = s_j}w'_j+\sum_{j \in  V_i^{h}: s_i = s_j}w'_j \leq D\big( \sum_{j \in V_i\setminus V_i^{h}: s_i = s_j}w''_j+\sum_{j \in  V_i^{h}: s_i = 
 s_j}w''_j \big)
 \end{equation}

 On the other hand for any $j$ and $j'$, if $j'$ belongs in a group lighter than $j$ then by construction $\frac{w''_{j}}{w''_{j'}}\geq \frac{n}{\eps}$ (recall the way the groups were brought closer), which gives $nw''_{j'}\leq\eps w''_j$, yielding
 \begin{equation}\label{eq:alg3}
 \sum_{j \in V_i\setminus V_i^{h}: s_i \neq s_j}w''_j+\sum_{j \in  V_i^{h}: s_i \neq s_j}w''_j\leq (1+\eps)\sum_{j \in  V_i^{h}: s_i \neq s_j}w''_j
 \end{equation}
 Using equations (\ref{eq:alg2}), (\ref{eq:alg1}) and (\ref{eq:alg3}), in this order, and that $D\sum_{j \in  V_i^{h}: s_i \neq  s_j}w''_j=\sum_{j \in  V_i^{h}: s_i \neq  s_j}w'_j \leq \sum_{j \in  V_i: s_i \neq s_j}w'_j  $,
 we get
 \[\sum_{j \in  V_i: s_i = s_j}w'_j \leq (1+\eps)^2\sum_{j \in V_i: s_i \neq s_j}w'_j\]
 as needed.
\end{proof}

\begin{remark}
We observe the following trade off: we can get a $(1+\eps)^2$-approximate equilibrium if we skip the rounding step at the beginning of the algorithm but then the number of different weights $D_\eps$ and thus the running time of the algorithm may increase. Also, if $\Delta$ is the maximum degree among the vertices of $G$, then replacing $\frac{n}{\eps}$ with $\frac{\Delta}{\eps}$ in the algorithm and following a similar analysis gives $O(\frac{m}{\eps}\lceil\frac{\Delta}{\eps}\rceil^{2D_\eps})$ running time.
\end{remark}

%% file: Appendix_Preliminaries.tex
\section{Missing Technical Details from the Proof of Theorem~\ref{t:main_max_cut}}
\label{s:MC}

In the following sections, we present all the technical details for the proof of Theorem~\ref{t:main_max_cut}. Recall that our \nmc{} instance is composed of the following gadgets:

\begin{enumerate}
    \item \emph{Leverage} gadgets that are used to transmit nonzero bias to vertices of high weight.
    \item Two \emph{Circuit Computing} gadgets that calculate the values and next neighbors of solutions.
    \item A \emph{Comparator} gadget.
    \item Two \emph{Copy} gadgets that transfer the solution of one circuit to the other, and vice versa.
    \item Two \emph{Control} gadgets that determine the (write or compute) mode in which the circuit operates. 
\end{enumerate}

Note that whenever we wish to have a vertex of higher weight that dominates all other vertices of lower weight, we multiply its weight with $2^{kN}$ for some constant $k$. We then choose $N$ sufficiently large so that, for all $k$, vertices of weight $2^{kN}$ dominate all vertices of weight $2^{(k-1)N}$. Henceforth, we will assume $N$ has been chosen sufficiently large.

In what follows, when we refer to the value of the circuit $C$, we mean the value that the underlying \cf{} instance $C$ would output, given the same input. Moreover, when we refer to the value of a vertex, we mean the side of the cut the vertex lies on. There are two values, $0$ and $1$, one for each side of the cut. 

As usual in previous work, we assume two \emph{supervertices}, a $1$-vertex and a $0$-vertex that share an edge and have a huge weight, which dominates the weight of any other vertex, e.g., $2^{1000N}$. As a result, at any local optimum, these vertices take complementary values. When we write, e.g., that $Flag = 1$ or $Control = 1$, we mean that \emph{Flag} or \emph{Control} takes the same value as the $1$-vertex. This convention is used throughout this section, always with the same interpretation. Our construction assumes that certain vertices always take a specific value, either $0$ or $1$. We can achieve this by connecting such vertices to the \emph{supervertex} of complementary value.

\subsection{Outline of the Proof of Theorem~\ref{t:main_max_cut}}
\label{sec:outline}

We start with a proof-sketch of our reduction, where we outline the key properties of our gadgets and the main technical claims used to establish the correctness of the reduction. 

At a high level, the proof of Theorem~\ref{t:main_max_cut} boils down to showing that at any local optimum of
the \lnmc{} instance of Figure~\ref{f:construction}, 
in which $\text{\emph{Flag}} =1$, the following hold:

\begin{enumerate}
    \item $\text{\emph{I}}_A= \text{\emph{Next}}B$
    \item $\text{\emph{Next}}B = \text{\emph{Real-Next}}(I_B)$
    \item $\text{\emph{Real-Val}}(I_B) \geq \text{\emph{Real-Val}}(I_A)$
\end{enumerate}

Once these claims are established, we can be sure that
the string defined by the values of vertices in $I_B$ defines a locally optimal solution for \cf{}. This is because the claims above directly imply that
$\text{\emph{Real-Val(I}}_B) \geq \text{\emph{Real-Val}} \lp(  \text{\emph{Real-Next}}\lp(I_B\rp)\rp)$, which means that there is no neighboring solution of $I_B$ of strictly better value. Obviously, we establish symmetrically the above claims when $\text{\emph{Flag}} =0$.

In the remainder of this section, we discuss the main technical claims required for the proof of Theorem~\ref{t:main_max_cut}. To do so, we follow a three step approach.  We first discuss the behavior of the \emph{Circuit Computing} gadgets $C_A$ and $C_B$. We then reason why $I_A = \text{\emph{Next}}B$, which we refer to as the \emph{Feedback Problem}. Finally we establish the last two claims, which we refer to as \emph{Correctness of the Outputs}.

\paragr{Circuit-Computing Gadgets.}
The \emph{Circuit Computing} gadgets $C_A$ and $C_B$ are
the basic primitives of our reduction. They are based on the gadgets introduced by Sch\"affer and Yannakakis \cite{SY91} to establish
PLS-completeness of \lmc{}. This type of gadgets can be constructed so as to simulate any Boolean circuit $C$.

The most important vertices are those corresponding to the input and the output of the simulated circuit $C$ and are denoted as $I,O$. Another important vertex is \emph{Control}, which allows the gadget to switch between the write and the compute mode of operation. Figure~\ref{f:CGadget} is an abstract depiction of the \emph{Circuit Computing} gadgets. 

The main properties of the gadget are described in Theorem~\ref{l:SY_gadgets}. Its proof is presented in
Section~\ref{s:computing_gadgets}, where the exact construction of the gadget is presented. We recall here the definition of bias, first discussed in Section~\ref{sec:nMCoverview}.

\begin{figure}[t]
    \centering
    \includegraphics[scale =0.7]{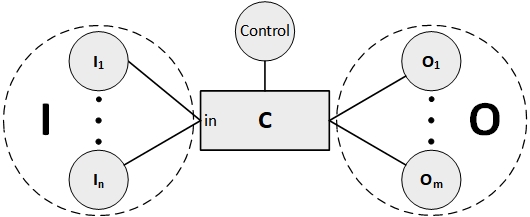}
    \caption{\emph{Circuit Computing} gadgets. The dashed circles, called $I$ and $O$, represent all input and output vertices, respectively. This type of (``hyper''-)vertex is represented in the rest of the figures with a bold border.}
\label{f:CGadget}
\end{figure}

\begin{definition}[Bias]\label{d:bias}
The \emph{bias} that a vertex $i$ experiences with respect to $V' \subseteq V$ is
\[\lp|\sum_{j \in V_i^1  \cap V'}w_j -\sum_{j \in V_i^0 \cap V' }w_j \rp|\,,\]
where $V_i^0$ (resp. $V_i^1$) is the set of neighbors of vertex $i$ on the $0$ (resp. $1$) side of the cut. 
\end{definition}

Bias is a key notion in the subsequent analysis. The gadgets presented in Figure~\ref{f:construction}
are a subset of the vertices of the overall instance. Each gadget is composed by the ``input vertices'', the internal vertices and the ``output vertices''. Moreover as we have seen
each gadget stands for a ``circuit'' with some specific functionality (computing, comparing, copying e.t.c.). Each gadget is specifically constructed so as at any local optimum of the overall instance, the output vertices of the gadget experience some bias towards some values that depend on the values of the input vertices of the gadget.
Since the output vertices of a gadget may also participate as input vertices at some other gadgets, it is important to quantify the bias of each gadget in order to prove consistency in our instance. Ideally, we would like to prove that at any local optimum the bias that a vertex experiences from a gadget in which it is an output vertex, is greater than the sum of the biases of the gadgets in which it participates as input vertex.

Theorem~\ref{l:SY_gadgets} describes the local optimum behavior of the input vertices $I_\ell$ and output vertices $Next\ell,Val_\ell$ of the $Circuit Computing$ gadgets $C_\ell$.

\begin{restatable}{theorem}{sygadgets}
\label{l:SY_gadgets}
\noindent
At any local optimum of the \lnmc{} of Figure~\ref{f:construction}.
\begin{enumerate}
    \item If $Control_\ell = 1$ and the vertices of $Next\ell$,$Val_\ell$ experience $0$ bias from any other gadget beyond $C_\ell$ then:
    \begin{itemize}
        \item $\text{\emph{Next}}\ell = \text{\emph{Real-Next}}(I_\ell)$
        \item $\text{\emph{Val}}_\ell = \text{\emph{Real-Val}}(I_\ell)$
     \end{itemize}
    \item If  $Control_\ell = 0$ then each vertex in $I_\ell$ experiences $0$ bias from the internal vertices of $C_\ell$.
    \item $Control_\ell$ experiences $w_{Control\ell}$ bias from the internal vertices of $C_\ell$.
\end{enumerate}
\end{restatable}

\medskip
Case~$1$ of Theorem~\ref{l:SY_gadgets}
describes the \emph{compute mode} of the \emph{Circuit Computing} gadgets.
At any local optimum with $ControlA =1$, and with the output vertices of $C_A$ being indifferent with respect to other gadgets, then $C_A$ computes its output correctly. Note that because the vertices in \emph{NextA},\emph{ValA} are also connected with internal vertices of other gadgets
(\emph{CopyA} and \emph{Comparator} gadgets)
that may create bias towards the opposite value, the second condition is indispensable. Case~$2$ of Theorem~\ref{l:SY_gadgets} describes the \emph{write mode}. If at a local optimum $ControlA = 0$ then the vertices in $I_A$ have $0$ bias from the $C_A$ gadget and as a result their value is determined by the biases of the \emph{CopyB} gadget and the \emph{Equality} gadget. Case~$3$ of Theorem~\ref{l:SY_gadgets} describes the minimum bias that the equality gadget must pose to the Control vertices
so as to make the computing gadget flip from one mode to the other. As we shall see, the weights $w_{ControlA} = w_{ControlB} = w_{Control}$ are selected much smaller than the bias the $Control\ell$ vertices experience due to
the \emph{Equality} gadgets, meaning that the \emph{Equality} gadgets control the \emph{write mode} and the \emph{compute mode} of the \emph{Circuit Computing} gadgets no matter the values of the vertices in $C_\ell$ gadgets.


\paragr{Solving the Feedback problem.}
Next, we establish the first of the claims above, i.e., at any local optimum of \lnmc{}
instance of Figure~\ref{f:construction}
in which $\text{\emph{Flag}} = 1$,
\emph{NextB} is written to $I_A$ and vice versa when $\text{\emph{Flag}} = 0$. This is formally stated in Theorem~\ref{l:write}.

\begin{theorem}
\label{l:write}
Let a local optimum of the \lnmc{} instance 
in Figure~\ref{f:construction}.
\begin{itemize}
    \item If $\text{\emph{Flag}} = 1$ then $I_A = \text{\emph{NextB}}$
    \item If $\text{\emph{Flag}} = 0$ then $I_B = \text{\emph{NextA}}$
\end{itemize}
\end{theorem}

We next present the necessary lemmas for proving Theorem~\ref{l:write}.

\begin{restatable}{lemma}{consistent}
\label{l:consistent1}
Let a local optimum of the \lnmc{} instance in Figure~\ref{f:construction}. Then,
\begin{itemize}
    \item $ControlA = (I_A  =  T_B)$
    \item $ControlB = (I_B  =  T_A)$
\end{itemize}
\end{restatable}

\medskip
In Section~\ref{s:control}, we present the construction of the \emph{Equality} gadget. This gadget is specifically designed so that at any local optimum, its internal vertices generate bias to
$ControlA$ towards the value of the predicate $(I_A  =  T_B)$.
Notice that if we multiply
all the internal vertices of the equality gadget with a positive constant, the bias $ControlA$ experiences
towards value $(I_A  =  T_B)$ is multiplied by the same constant (see Definition~\ref{d:bias}). Lemma~\ref{l:consistent1} is established
by multiplying these weights
with a  sufficiently large
constant so as to make this bias larger than $w_{ControlA}$.
We remind that by Theorem~\ref{l:SY_gadgets}, the bias that $ControlA$ experiences from $C_A$ is $w_{ControlA}$. As a result, the local optimum value of $ControlA$ is $(I_A  =  NextB)$ no matter the values of $ControlA$ 's neighbors in the $C_1$ gadget.
The \emph{red mark} between $ControlA$ and the $C_A$ gadget in Figure~\ref{f:construction} denotes the ``indifference'' of $ControlA$ towards the values of the $C_A$ gadget (respectively for $ControlB$).

In the high level description of the
\lnmc{} instance of Figure~\ref{f:construction}, when $\text{\emph{Flag}} = 1$
the values of \emph{NextB} is copied to $I_A$ as follows: At first $T_B$ takes the value of \emph{NextB}. If $I_A \neq T_B$ then $ControlA = 0$
and the $C_A$ gadget switches to \emph{write mode}. Then the vertices in $I_A$ takes the values of the vertices in \emph{NextB}. This is formally stated in Lemma~\ref{l:copy2}.

\begin{restatable}{lemma}{copyB}
\label{l:copy2}
At any local optimum point of the \lnmc{} instance of Figure~\ref{f:construction}:
\begin{itemize}
\item If $\text{\emph{Flag}} = 1$, i.e. \emph{NextB} \emph{writes} on $I_A$, then
\begin{enumerate}
\item $T_B = \text{\emph{NextB}}$
\item If $ControlA = 0$ then $I_A = T_B = \text{\emph{NextB}}$
\end{enumerate}

\item If $\text{\emph{Flag}} = 0$ i.e. \emph{NextA} \emph{writes} on $I_B$, then
\begin{enumerate}
\item $T_A = \text{\emph{NextA}}$
\item If $ControlB = 0$ then $I_B = T_A = \text{\emph{NextA}}$
\end{enumerate}
\end{itemize}
\end{restatable}

\medskip 
In Section~\ref{s:copy}, we present the construction of the \emph{Copy} gadgets. At a local optimum where $\text{\emph{Flag}} = 1$, this gadget creates bias to the vertices in $I_A,T_B$ vertices towards adopting the values of \emph{NextB}. Since $I_A,T_B$
also participate in the \emph{Equality} gadget in order to establish Lemma~\ref{l:copy2} we want to make the bias of the \emph{CopyB} gadget larger than the bias of the \emph{Equality} gadget. This is done by again by multiplying the weights of the internal vertices of \emph{CopyB} with a sufficiently large constant. The ``indifference'' of the vertices in $I_A,T_B$ with respect to the values of the internal vertices of the \emph{Equality} gadget is denoted
in Figure~\ref{f:construction}
by the \emph{red marks} between the vertices in $I_A,T_B$ and the \emph{Equality} gadget.

In Case~$2$ of Lemma~\ref{l:copy2} the additional condition $ControlA = 0$ is necessary to ensure that $I_A = \text{\emph{NextB}}$. The reason is that the bias of the \emph{Copy} gadget to the vertices in $I_A$ is sufficiently larger than the bias of the \emph{Equality} gadget to the vertices in $I_A$, but not necessarily to the bias of the $C_A$ gadget. The condition $ControlA = 0 $ ensures $0$ bias of the $C_A$ gadget to the vertices $I_A$, by Theorem~\ref{l:SY_gadgets}. As a result the values of the vertices in $I_A$ are determined by the values of their neighbors in the \emph{CopyB} gadget.

\begin{proof}[Proof of Theorem~\ref{l:write}]
Let a local optimum in which $\text{\emph{Flag}} = 1$.
Let us assume that $I_A \neq \text{\emph{NextB}}$.
By Case~$1$ of Lemma~\ref{l:copy2}, $T_B = \text{\emph{NextB}}$. As a result, $I_A \neq T_B$, implying that $ControlA = 0$ (Lemma~\ref{l:consistent1}). Now, by Case~$2$ of Lemma~\ref{l:copy2} we have that $I_A = \text{\emph{NextB}}$, which is a contradiction. The exact same analysis holds when $\text{\emph{Flag}} = 0$.
\end{proof}

\paragr{Correctness of the Output Vertices.}
In the previous section we discussed how the \emph{Feedback problem} ($I_A = \text{\emph{NextB}}$ when $\text{\emph{Flag}} = 1$) is solved in our reduction.
We now exhibit how the two last cases of our initial claim are established.

\begin{theorem}\label{t:2}
At any local optimum of the instance of \lnmc{} of
Figure~\ref{f:construction}:
\begin{itemize}
\item If $\text{\emph{Flag}} = 1$
    \begin{enumerate}
    \item $\text{\emph{Real-Val}}(I_A) \leq \text{\emph{Real-Val}}(I_B)$
    \item $\text{\emph{NextB}} = \text{\emph{Real-Next}}(I_B)$
\end{enumerate}

\item If $\text{\emph{Flag}} = 0$
    \begin{enumerate}
    \item $\text{\emph{Real-Val}}(I_B) \leq \text{\emph{Real-Val}}(I_A)$
    \item $\text{\emph{NextA}} = \text{\emph{Real-Next}}(I_A)$
\end{enumerate}
\end{itemize}
\end{theorem}

At first we briefly explain the difficulties in establishing Theorem~\ref{t:2}. In the following discussion we assume that $\text{\emph{Flag}} = 1$, since everything we mention holds symmetrically for $\text{\emph{Flag}} = 0$. Observe
that if $\text{\emph{Flag}} = 1$ we know nothing about the value of $ControlB$ and as a result we cannot guarantee that
$\text{\emph{NextB}}  = \text{\emph{Real-Next}} (I_B)$ or $\text{\emph{Val}}_B = \text{\emph{Real-Val}}(I_B)$. But even in the case of $C_A$ where $ControlA = 1$ due to Theorem~\ref{l:write},
the correctness of the vertices in $\text{\emph{NextA}}$ or $\text{\emph{Val}}_A$ cannot be guaranteed.
The reason is that in order to apply Theorem~\ref{l:SY_gadgets}, \emph{NextA} and $\text{\emph{Val}}_A$ should experience $0$ bias with respect to
any other gadget they are connected to. But at a local optimum, these vertices may select their values according to the values of their \emph{heavily weighted neighbors} in the \emph{CopyA} and the \emph{Comparator} gadget.

The correctness of the values of the output vertices, i.e. $\text{\emph{NextA}} = \text{\emph{Real-Next}}(I_A)$ and $\text{\emph{Val}}_A = \text{\emph{Real-Val}}(I_A)$, is ensured
by the design of the \emph{CopyA} and the \emph{Comparator} gadgets. Apart from their primary role these gadgets are specifically designed to
cause $0$ bias to the output vertices of the \emph{Circuit Computing} gadget to which
the better neighbor solution is written. In other words at any local optimum in which $\text{\emph{Flag}} = 1$ and any vertex in $C_A$:
the total weight of its neighbors (belonging in the \emph{CopyA} or the \emph{Comparator} gadget) with value $1$
equals the total weight of its neighbors (belonging in the \emph{CopyA} or the \emph{Comparator} gadget) with value $0$.

The latter fact is denoted by the \emph{green marks} in Figure~\ref{f:construction_1} and permits the application of Case~$1$ of Theorem~\ref{l:SY_gadgets}.
Lemma~\ref{l:copy_unbias} and ~\ref{l:comparator_unbias} formally state these ``green marks''.

\begin{figure}
    \centering
    \includegraphics[scale= 0.5]{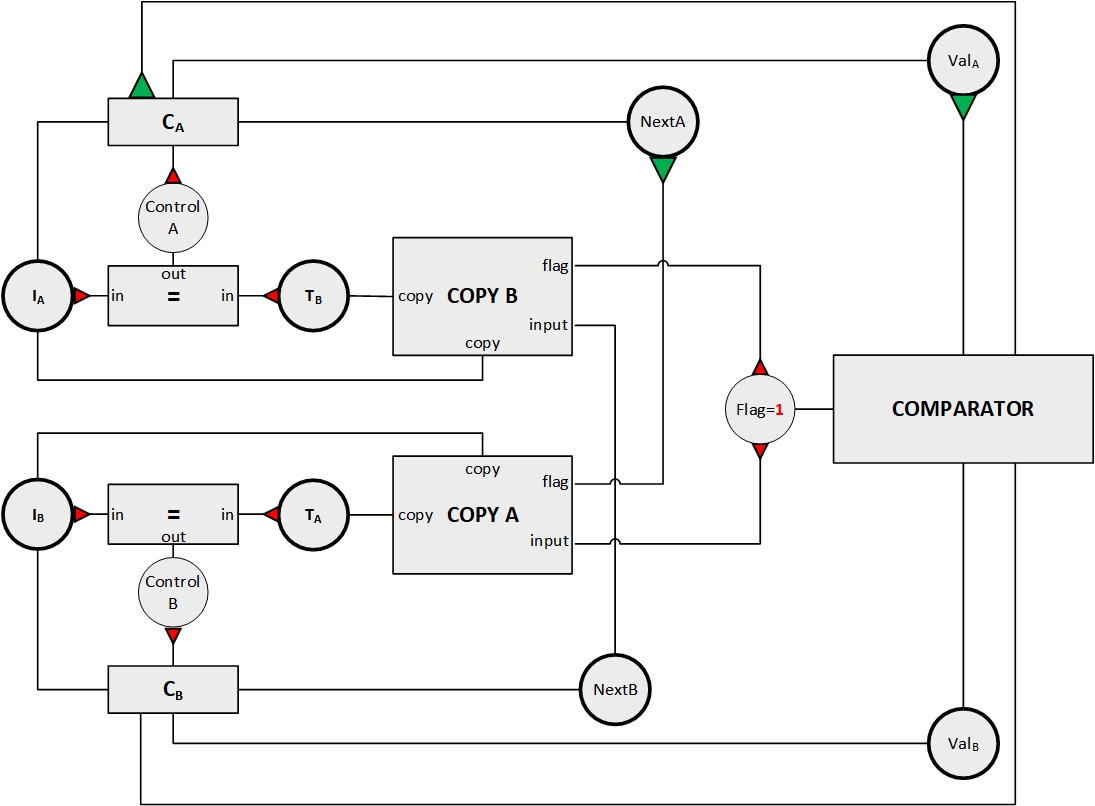}
    \caption{Since $Flag = 1$, any internal vertex of the $C_B$ gadget has $0$ bias with respect to all the other gadgets. As a result, Theorem~\ref{l:SY_gadgets} applies.}
\label{f:construction_1}
\end{figure}

\begin{restatable}{lemma}{copyUnbias}
\label{l:copy_unbias}
At any local optimum of the \lnmc{} instance of Figure~\ref{f:construction}:
\begin{itemize}
    \item If $\text{\emph{Flag}} = 1$, then any vertex in \emph{NextA} experience $0$ bias with respect to the \emph{CopyA} gadget.
\item If $\text{\emph{Flag}} = 0$ then
then any vertex in \emph{NextB} experience $0$ bias wrt. the \emph{CopyB} gadget.
\end{itemize}
\end{restatable}

\begin{restatable}{lemma}{comparatorUnbias}
\label{l:comparator_unbias}
Let a local optimum of the instance of \lnmc{} of Figure~\ref{f:construction}:
\begin{itemize}
    \item If \emph{Flag} = 1 then all vertices of $C_A$ experience $0$ bias wrt. the \emph{Comparator} gadget.

    \item If \emph{Flag} = 0 then all vertices of $C_B$ experience $0$ bias wrt. the \emph{Comparator} gadget.
\end{itemize}
\end{restatable}

\begin{remark}
The reason that in Lemma~\ref{l:comparator_unbias} we refer to all vertices of $C_A$ (respectively $C_B$) and not just to the vertices in $Val_A$ (respectively $Val_B$) is that in the constructed instance of \lnmc{} of Figure~\ref{f:construction}, we connect internal vertices of the $C_A$ gadget with internal vertices of the \emph{Comparator} gadget. This is the only point in our construction where internal vertices of different gadgets share an edge and is denoted in Figure~\ref{f:construction} and~\ref{f:construction_1} with the direct edge between the $C_A$ gadget and the \emph{Comparator} gadget.
\end{remark}

Now using Lemma~\ref{l:copy_unbias} and Lemma~\ref{l:comparator_unbias} we can prove the correctness of the output vertices $\text{\emph{NextA}}, \text{\emph{Val}}_A$ when $\text{\emph{Flag}} =1$ i.e.
$\text{\emph{NextA}} = \text{\emph{Real-Next}}(I_A)$ and $\text{\emph{Val}}_A = \text{\emph{Real-Val}}(I_B)$ (symmetrically
for the vertices in
$\text{\emph{NextB}}, \text{\emph{Val}}_B$ when
$\text{\emph{Flag}} = 0$).

\begin{restatable}{lemma}{correctnessA}
\label{l:correctness1}
Let a local optimum of the instance of \lnmc{} of Figure~\ref{f:construction}:
\begin{itemize}
    \item If $\text{\emph{Flag}} = 1$ then $\text{\emph{NextA}} = \text{\emph{Real-Next}}(I_A)$, $\text{\emph{Val}}_A =  \text{\emph{Real-Val}}(I_A)$.
    \item If $\text{\emph{Flag}} = 0$ then $\text{\emph{NextB}} = \text{\emph{Real-Next}}(I_B)$, $\text{\emph{Val}}_B =  \text{\emph{Real-Val}}(I_B)$.
\end{itemize}
\end{restatable}

\begin{proof}
We assume that $\text{\emph{Flag}} = 1$ (for $\text{\emph{Flag}} = 0$ the exact same arguments hold). By Theorem~\ref{l:write} we have $I_A = \text{\emph{NextB}}$ and by Lemma~\ref{l:copy2} we have that $T_B = \text{\emph{NextB}}$. As a result, $I_A = T_B$ and by Lemma~\ref{l:consistent1}
$ControlA = 1$.
Lemma~\ref{l:copy_unbias} and Lemma~\ref{l:comparator_unbias} guarantee that the vertices in $\text{\emph{NextA}},\text{\emph{Val}}_A$ of $C_A$ experience $0$ bias towards all the other gadgets of the construction and since $ControlA = 1$,
we can apply Case~$1$ of Theorem~\ref{l:SY_gadgets} i.e.
$\text{\emph{Val}}_A = \text{\emph{Real-Val}}(I_A)$ and
$\text{\emph{NextA}}= \text{\emph{Real-Next}}(I_A)$.
\end{proof}

Up next we deal with the correctness of the values of the output vertices in $\text{\emph{Val}}_B$ and $\text{\emph{NextB}}$ when $\text{\emph{Flag}} = 1$.
We remind again that, even if at a local optimum $ControlB=1$, we could not be sure about the correctness
of the values of these output vertices due to the bias their neighbors in the \emph{CopyB} and the \emph{Comparator} gadget (Theorem~\ref{l:SY_gadgets} does not apply).
The \emph{Comparator} gadget plays a crucial role in solving this last problem.
Namely, it also checks whether the output vertices in \emph{NextB} have correct values with respect to the input $I_B$ and
if it detects incorrectness it outputs $0$. This is done by
the connection of some specific internal vertices
of the $C_A,C_B$ gadgets with the internal vertices of the \emph{Comparator} gadget (Figure~\ref{f:construction}: edges between $C_A,C_B$ and \emph{Comparator}).

\begin{restatable}{lemma}{comparatorCorrectness}
\label{l:comparator_correctness}
At any local optimum of the \nmc{} instance of Figure~\ref{f:construction}:

\begin{itemize}
    \item If $\text{\emph{Flag}} = 1$ then $\text{\emph{NextB}} = \text{\emph{Real-Next}}(I_B)$

    \item If $\text{\emph{Flag}} = 0$
     then $\text{\emph{NextA}} = \text{\emph{Real-Next}}(I_A)$
\end{itemize}
\end{restatable}

\medskip
We highlight that the correctness of values of the output vertices \emph{NextB}, i.e. $\text{\emph{NextB}} = \text{\emph{Real-Next}}(I_B)$, is not guaranteed by application of Theorem~\ref{l:SY_gadgets} (as in the case of correctness of $NextA,Val_A$), but from the construction of the \emph{Comparator} gadget. Lemma~\ref{l:comparator_correctness} is proven in Section~\ref{s:comparison_gadget} where the exact construction of this gadget is presented. Notice that Lemma~\ref{l:comparator_correctness} says nothing about the correctness of the values of the output vertices in $Val_B$. As we latter explain this cannot be guaranteed in our construction. Surprisingly enough, the \emph{Comparator}
outputs the right outcome of the predicate $\lp(\text{\emph{Real-Val}}(I_A) \leq \text{\emph{Real-Val}}(I_B)\rp)$ even if
$\text{\emph{Val}}_B \neq \text{\emph{Real-Val}}(I_B)$. The latter is one of our main technical contributions in the reduction that reveals the difficulty of \lnmc{}. The crucial differences between our \emph{Comparator} and the \emph{Comparator} of the previous reductions \cite{SY91,GS10,ET11} are discussed in the end of the section. Lemma~\ref{l:comparator} formally states the robustness of the outcome of the \emph{Comparator} even with ``wrong values''
in the vertices of $Val_B$ and is proven in Section~\ref{s:comparison_gadget}.

\begin{restatable}{lemma}{comparator}
\label{l:comparator}
At any local optimum of the \lnmc{} instance of Figure~\ref{f:construction}:
\begin{itemize}
    \item If $Flag = 1$, $\text{\emph{NextA}} = \text{\emph{Real-Next}}(I_A)$,
    $\text{\emph{Val}}_A = \text{\emph{Real-Val}}(I_A)$
    and $\text{\emph{NextB}} = \text{\emph{Real-Next}}(I_B)$
    then $Real\text{-}Val(I_A) \leq Real\text{-}Val(I_B)$.
    \item If $Flag = 0$, $\text{\emph{NextB}} = \text{\emph{Real-Next}}(I_B)$,
    $\text{\emph{Val}}_B = \text{\emph{Real-Val}}(I_B)$
    and $\text{\emph{NextA}} = \text{\emph{Real-Next}}(I_A)$
    then $Real\text{-}Val(I_B) \leq Real\text{-}Val(I_A)$.
\end{itemize}
\end{restatable}

\medskip
We are now ready to prove Theorem~\ref{t:2}.

\begin{proof}[Proof of Theorem~\ref{t:2}]
Let a local optimum of the instance of Figure~\ref{f:construction} with $\text{\emph{Flag}} = 1$ (respectively for $\text{\emph{Flag}} = 0$).
By Lemma~\ref{l:comparator_correctness}, $\text{\emph{NextB}} = \text{\emph{Real-Next}}(I_B)$ and thus Case~$1$ is established.
Moreover by Lemma~\ref{l:correctness1}, $\text{\emph{NextA}} = \text{\emph{Real-Next}}(I_A)$ and $\text{\emph{Val}}_A = \text{\emph{Real-Val}}(I_A)$. As a result,
Lemma~\ref{l:comparator} applies and $\text{\emph{Real-Val}}(I_A) \leq \text{\emph{Real-Val}}(I_B)$ (Case~$2$ of Theorem~\ref{t:2})
\end{proof}

\paragr{Putting Everything Together.}
Having established Theorem~\ref{l:write} and \ref{t:2}, the $PLS$-completeness of \lnmc{} follows easily. For the sake of completeness, we 
show how we can put everything together and conclude the proof Theorem~\ref{t:main_max_cut}.

\begin{proof}[Proof of Theorem~\ref{t:main_max_cut}]
For a given circuit $C$ of the \cf{}, we can construct in polynomial time the instance of \lnmc{} of Figure~\ref{f:construction}.
Let a local optimum of this instance. Without loss of generality, we assume that $Flag = 1$.
Then, by Theorem~\ref{l:write} and Theorem~\ref{t:2}, $I_A = NextB$, $NextB = Real-Next(I_B)$ and $Real-Val(I_A) \leq Real-Val(I_B)$. Hence, we have that
\begin{eqnarray*}
Real-Val(I_B) &\geq& Real-Val(I_A)\\
&=& Real-Val(NextB)\\
&=& Real-Val\lp( Real-Next(I_B) \rp)
\end{eqnarray*}
But if $I_B \neq Real-Next(I_B)$, then $Real-Val(I_B)  > Real-Val\lp( Real-Next(I_B) \rp)$ which is a contradiction.
Thus $I_B = Real-Next(I_B) = I_B$, meaning that the string defined by the values of $I_B$ is a locally optimal solution for the \cf{} problem.
\end{proof}

In the following sections, we present the gadget constructions in detail and all the formal proofs missing from this section.

%% file: Appendix_Leverage.tex
\subsection{The Leverage Gadget} \label{s:leverage}

The \emph{Leverage} gadget is a basic construction in the PLS completeness proof. 
This gadget solves a basic problem in the reduction. Suppose that we have a vertex with 
relatively small weight $A$ and we want to bias a vertex with large weight $B$. For example,
the large vertex might be indifferent towards its other neighbors, which would allow even a
small bias from the small vertex to change its state. We would also like to ensure that
the large vertex does not bias the smaller one with very large weight, in order for the 
smaller to retain its value. 

This problem arises in various parts of the PLS proof. For example, we would like the outputs of
a circuit to be fed back to the inputs of the other one. The outputs have very small weight 
compared to the inputs, since the weights drop exponentially in the \emph{Circuit Computing} gadget. 
We would like the inputs of circuit $B$ to change according to the outputs of circuit $A$
and not the other way around. Another example involves the \emph{Equality} Gadget, which influences
the $\text{\emph{Control}}_\ell$ of the \emph{Circuit Computing} gadget. The vertices of the Gadget 
have weights of the order of
$2^{10N}$, while the control vertices of the \emph{Circuit Computing} gadget are of the order of $2^{100N}$. 
We would like the output of the gadget to bias the $\text{\emph{Control}}_\ell$ vertices, while also
remaining independent from them. 

Let's get back to the original problem. A naive solution would be to connect vertex $A$ directly 
with vertex $B$. However, this would result in vertex $B$ biasing vertex $A$ due to the larger
weight it possesses. For example, if we connected $Control_A$ with the control variables of
circuit $B$, then they would always bias $Control_A$ with a very large weight, rendering
the entire \emph{Equality} gadget useless. We would like to ensure that vertex $A$ biases $B$ with
a relatively small weight, while also experiencing a small bias from it. 

The solution we propose is a \emph{Leveraging} gadget that is connected between vertices
$A$ and $B$. It's construction will depend on the weights $A$ and $B$, as well as the 
bias that we would like $B$ to experience from $A$. Before describing the construction, we
discuss it's functionality on a high level. 

As shown in Figure \ref{fig:leverage}, we place the gadget between the vertices $A$ and $B$. We use two parameters
$x$, $\eps$ in the construction. 
We first want to ensure that vertex $A$ experiences a small bias from the gadget. This is why
we put vertices $L_{1,1},L_{1,2}$ at the start with weight $B/2^{x+1} + \eps$, which puts a 
relatively small bias. We want these vertices to be dominated by $A$. This is why vertices
$L_{1,3}$, $L_{1,4}$ have combined weight less than $A$. However, these vertices cannot directly 
influence $B$, since it's weight dominates the weights of $L_{1,1}$, $L_{1.2}$. For this reason,
we repeat this construction $x+1$ times, until vertices $L_{x,1}$, $L_{x,2}$, whose combined weight
is slightly larger than $B$. This means that vertices $L_{x,3}$, $L_{x,4}$ are not dominated by
$B$ and can therefore be connected directly with it. The details of the proof are given below.

\begin{figure}[t]
    \centering
    \includegraphics[width=\textwidth]{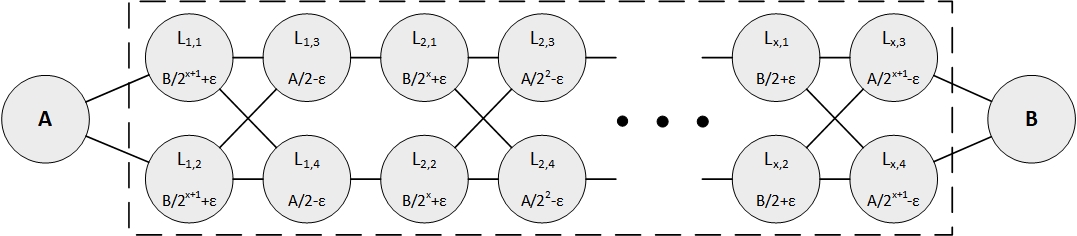}
    \caption{The \emph{Leverage} Gadget.}
\label{fig:leverage}
\end{figure}

\begin{restatable}{lemma}{leverage}
\label{l:leverage}
If the input vertex $A$ of a leverage gadget with output vertex $B$, parameters $x,\eps$, has value 1, then the output vertex experiences bias $w_A/{2^x}+2\eps$ towards 0, while the input vertex $A$ experiences bias $w_B/{2^x}-2\eps$ towards 1. If $A$ has value 0, then $B$ experiences the same bias towards 1, while $A$ is biased towards 0.   
\end{restatable}

\begin{proof}
We first consider the vertices $L_{1,1},L_{1,2}$. They both experience bias $w_A$ towards the opposite value of $A$, which is greater than the remaining weight of their neighbors $2w_A-2\eps$, and hence they are both dominated to take the opposite value of $A$. Similarly, the vertices $L_{1,3},L_{1,4}$ are now biased to take the opposite values of $L_{2,1},L_{2,2}$ with bias at least $w_B/2^{x}+2\eps$, which is greater than the remaining neighbors of $w_B/2^{x}+\eps$. Hence, both $L_{1,3},L_{1,4}$ have the same value as $A$ in any local optimum. In a similar way, we can prove that, in any local optimum $L_{i,3}=L_{i,4}=A$, and therefore $B$ experiences bias $w_A/{2^x}+2\eps$ towards the opposite value of A, while $A$ experiences bias at most $w_B/{2^x}-2\eps$ from this gadget.  
\end{proof}

Note that the above lemma works for any value of $\eps$. This means that we can make the bias that $B$ experiences arbitrarily close to $w_B/{2^x}$. For all cases where such a \emph{Leverage} gadget is used, it is implied that $\eps = 2^{-1000N}$ which is smaller than all other weights in the construction. Hence, we only explicitly specify the $x$ parameter and, for simplicity, such a \emph{Leverage} gadget is denoted as below schematically.

\begin{figure}[t]
    \centering
    \includegraphics[scale=0.7]{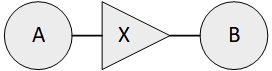}
    \caption{\emph{Leveraging Gadget} notation}
\label{fig:leveragenotation}
\end{figure}

%% file: Appendix_ComputeCircuit.tex
\subsection{The Circuit Computing Gadget}\label{s:computing_gadgets}

Each of the two computing circuits is meant to both calculate the value of the underlying \cf{} instance, as well as the best neighboring solution. For technical reasons one of the two circuits will need to output the complement of the value instead of the value itself, so that comparison can be achieved later with a single vertex.

In this section we present the gadgets that implement the above circuits in a \lnmc{} instance. The construction below is similar to the constructions of Sch\"affer and Yannakakis used to prove \lmc{} PLS-complete \cite{SY91}. Since NOR is functionally complete we can implement any circuit with a combination of NOR gates. In particular, each NOR gate is composed of the gadgets below. Each such gadget is parameterized by a variable $n$, and a NOR gadget with parameter $n$ is denoted NOR($n$). Since we wish for earlier gates to dominate later gates we order the gates in reverse topological order, so as to never have a higher numbered gate depend on a lower numbered gate. The $i$th gate in this ordering corresponds to a gadget $NOR(2^{N+i})$. Note that the first gates of the circuit have high indices, while the final gates have the least indices.  

We take care to number the gates so that the gates that each output the final bit of the value of the circuit are numbered with the $n$ lowest indexes, i.e. the gate of the $k$th bit of the value corresponds to a gate $NOR(2^{N+k})$. This is necessary so that their output vertices can be used for comparing the binary values of the outputs.

The input vertices of these gadgets are either an input vertex to the whole circuit or they are the output vertex of another NOR gate, in which case they have the weight prescribed by the previous NOR gate. The input vertices of the entire circuit (which are not the output vertices of any NOR gate) are given weight $2^{5N}$. 

\begin{figure}[t]
    \centering
    \includegraphics[scale=0.8]{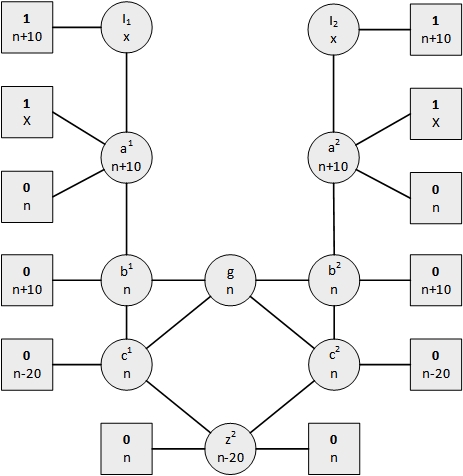}
    \,\,\,\,\,\,\,\,\,\,\,
    \includegraphics[scale=0.8]{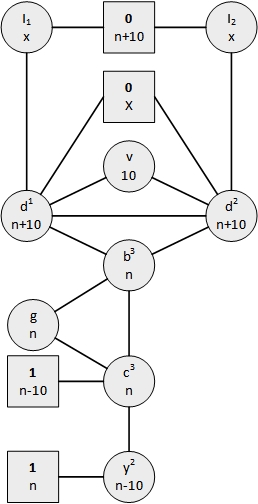}
    \caption{The \nmc{} instance implementing a NOR($n$) gadget.}
\label{f:SY_gadgets}
\end{figure}

Moreover, we have $y_i^1$,$z_i^1$ vertices which are meant to bias the internal vertices of each gadget and determine its functionality. Specifically, $a_i^1,a_i^2,c_i^1,c_i^2,v_i,b_i^3$ are biased to have the same value as $y_i^1$, while $b_i^1,b_i^2,d_i^1,d_i^2,c_i^3$ are biased to have the same value as $z_i^1$. This is achieved by auxiliary vertices of weight $2^{-200N}$, shown in Figure~\ref{f:localbias}.

\begin{figure}[t]
    \centering
    \includegraphics[scale=0.7]{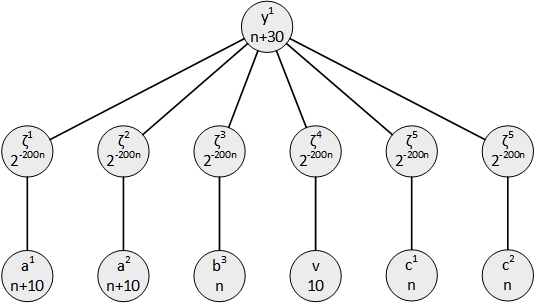}
    \,\,\,\,\,\,\,\,\,\,\,
    \includegraphics[scale=0.7]{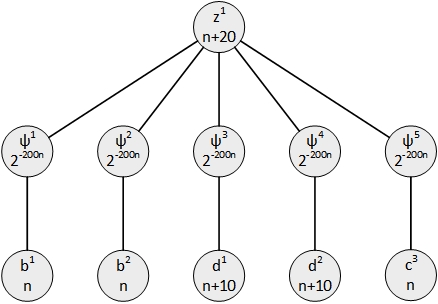}
    \caption{Local bias to internal vertices from $y_i^1,z_i^1$}
\label{f:localbias}
\end{figure}

We also have auxiliary vertices $\rho$ of weight $2^{-500N}$ that bias the output vertex $g_i$ to the correct NOR output value. Note that these vertices have the lowest weight in the entire construction.

\begin{figure}[t]
    \centering
    \includegraphics[scale=0.86]{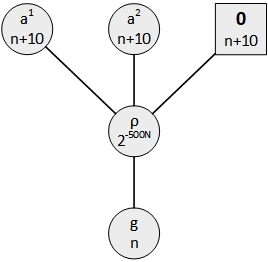}
    \caption{Extremely small bias to NOR output value.}
\label{f:norbias}
\end{figure}

These control vertices, $y^1,z^1,y^2,z^2$ are meant to decide the functionality of the gadget. We say that the $y,z$ vertices have their natural value when $y=1$ and $z=0$. We say they have their unnatural value when $y=0$ or $z=0$. In general, when these vertices all have their natural values the NOR gadget is calculating correctly and when they have their unnatural values the circuit's inputs are indifferent to the gadget. 

Unlike Sch\"affer and Yannakakis \cite{SY91}, we add two extra control variable vertices $y^3,z^3$ to each such NOR gadget, both of weight $n-50$. The reason is to ascertain that in case of incorrect calculation at least one $y$ variable will have its unnatural value. Otherwise, it would be possible, for example, to have an incorrect calculation with only $z^2$ being in an unnatural state.  

These NOR gadgets are not used in isolation, but instead compose a larger computing circuit. As Sch\"affer and Yannakakis do (\cite{SY91}), we connect each of the control variables $z^i,y^i$ of the above construction so as to propagate their natural or unnatural values depending on the situation. The connection of these gadgets is done according to the ordering we established earlier. Recall that the last $m$ gates correspond to gadgets calculating the value bits, the $n$ gates before them correspond to the output gates of the next neighbor, and the rest are internal gates of the circuit. 

\begin{figure}[t]
    \centering
    \includegraphics[width=\textwidth]{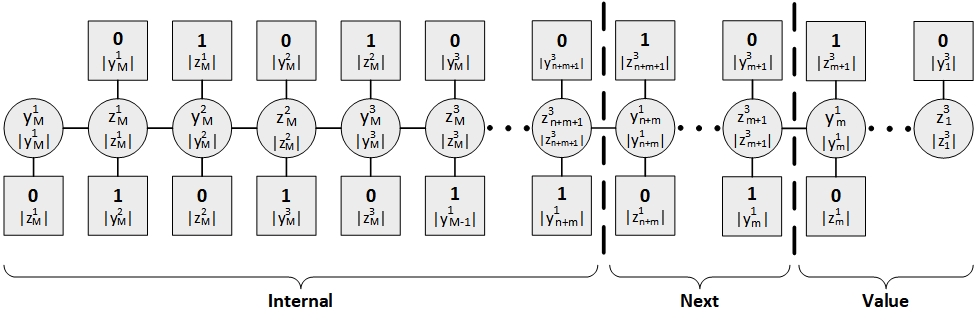}
    \caption{Connecting the control vertices of the NOR gadgets. Recall that $M$ is the number of total gates in the circuit, $n$ is the number of solution bits and $m$ is the number of value bits. Note that the gates are ordered in reverse, i.e the first gates have highest index.}
\label{f:chain}
\end{figure}

These gadgets' function is twofold. Firstly, they detect a potential error in a NOR calculation and propagate it to further gates, if the control variables have their unnatural values. Second,if the control variables have their unnatural values, they insulate the inputs so that they are indifferent to the gadget and can be changed by any external slight bias.

Furthermore, all the vertices of these gadgets are all multiplied by a $2^{100N}$ weight, except the vertices of the NOR gadget corresponding to the final bits of the value which are multiplied by $2^{90N}$. This is so that a possible error in the calculation of the next best neighbor supersedes any possible result of the comparison. The auxiliary vertices introduced above, which are meant to induce small biases to internal vertices, are not multiplied by anything.

Lastly, for technical simplicity, we have a single vertex for each computing circuit meant to induce bias to all control variable vertices $y,z$ at the same time. The topology of the connection is presented below.

\begin{figure}[t]
    \centering
    \includegraphics[scale=0.7]{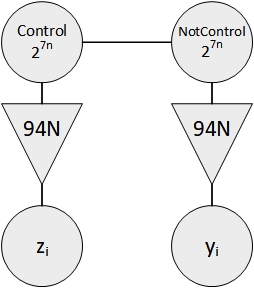}
    \caption{We use a single vertex Control to bias all control vertices $y,z$. Note that this vertex is connected with the $y,z$ vertices through leverage gadgets}
\label{f:control}
\end{figure}

We now prove the properties of these gadgets.

\begin{lemma}\label{l:SY_indifferency}
At any local optimum, if $z_i^1=1$ and $y_i^1=0$, then $I_1(g_i)$, $I_2(g_i)$ are indifferent with respect to the gadget $G_i$.
\end{lemma}

\begin{proof}
Since $z_i^1=1$ and $y_i^1=0$, by the previous lemma, $z_i^2=1$ and $y_i^2=0$
Since $y_i^1=0$, $a_i^1,a_i^2,c_i^1,c_i^2,v,b_i^3$ have an $\epsilon = 2^{-200N}$ bias towards 0. 
Since $z_i^1=1$, $d,c_i^3,b_i^1,b_i^2$ have an $\epsilon = 2^{-200N}$ bias towards 1. 
Assume $g_i=0$. Then $b_i^1$ has bias at least $2^{100N}\cdot(2\cdot2^i+10)+2^{-200N}$ towards 1 
which dominates his best response. Hence, $b_i^1=1$ in this case. 
Now $a_i^1$ has bias at least $w(I1)+2^{100N}\cdot2^i+2^{-200N}$ towards 0 which also dominates. Therefore, 
$a_i^1=0$. Similarly, $a_i^2=0$. 
Moreover, $c_i^3$ has $y_i^2$ and $g$ as neighbors which are both 0 so it can take
its preferred value of $c_i^3=1$.
Assume both $d_i^1,d_i^2$ are 0. Then $v=1$ and $b_i^3=1$  and hence at least one of $d_i^1,d_i^2$ would have incentive to change to 1. If $d_i^1=0,d_i^2=1$ then $v=0$ due to its $2^{-200N}$ bias. Also, $b_i^3=0$ because $d_i^1=0$, $d_i^2=1$, $g=0$, $c_i^3=1$ balance each other out $b_i^3=0$ due to its $2^{-200N}$ bias. Since $d_i^1$ experiences at least $2^{100N}\cdot{2\cdot2^i+10}+w(I_i^1)+2^{-200N}$ bias towards 1 it can only be $d_i^1=1$ in local optimum.
Hence, if $g=0 \implies d_i^1=d_i^2=1,a_i^1=a_i^2=0$.
Assume $g_i=1$. Then $c_i^1$ experiences bias towards 0 from $g_i$ and $z_i^2$ which together with the $2^{-200N}$ bias from $y_i^1$
means that his dominant strategy is to take the value 0.
Now $b_i^1$ experiences bias from $c_i^1=0$ towards 1 as well as bias $2^{100N}\cdot{2^i+10}$ towards 1. Along with the $2^{-200N}$ bias from
$z_i^1$ we have that $b_i^1=1$ in any local optimum. Similarly, $b_i^2=1$ by symmetry.
Hence, in this case as well $a_i^1$ is 0 in any local optimum. Similarly, we get that $a_i^2=0$ in any local optimum.

Assume both $d_i^1=d_i^2=0$ then, as above, we have that $b_i^3=1$ and hence at least one of the $d$ would gain the edge of weight $2^{-200N}$ by taking the value 1. Hence, at least one $d$ is equal to 1 and $b_i^3=0$ since it is indifferent with respect to $g,c_i^3$. Since $v$ is now indifferent with respect to $d_i^1=0,d_i^2=1$ it takes its preferred value $v=0$. Since $b_i^3=v=0$ we have that both $d_i^1,d_i^2$ must take their preferred values $d_i^1=d_i^2=1$ in any local optimum. 

In both cases both $a_i^1=a_i^2=0, d_i^1=d_i^2=1$ and hence $I_1(g_i),I_2(g_i)$ are indifferent with respect to the gadget.
\end{proof}

\begin{lemma}\label{l:SY_detection}
If gate $G_i$ is incorrect, then $z_i^2 = 1$. If $y_i^2 = 0$ then $z_i^2 = 1$. 
If $z_i^2 = 1$, then for all $j < i$ $z_j^1 = z_j^2 = z_j^3 = 1$ and 
$y_j^1 = y_j^2 = y_j^3 = 0$.
\end{lemma}
\begin{proof}
There are two possibilities if $G_i$ is incorrect. 
Either one of the inputs $I_1(g_i)$, $I_2(g_i)$ is $1$ and $g_i = 1$
or both $I_1(g_i) = I_2(g_i) = 0$
and $g_i = 0$.

In the first case, without loss of generality we have that $I_1(g_i) = 1$. 
This means that vertex $a_i^1$ is biased
towards value $0$ with weight at least $2\cdot2^{i+1}\cdot2^{100N}$ by $I_1(g_i)$ and
constant vertex $1$. This bias is greater than the weight of all the other neighbors
of $a_i^1$ combined. Hence, in local optimum, $a_i^1 = 0$. Hence, vertex $b_i^1$ is
biased towards value $1$ with weight at least $2\cdot|a_i^1|$, which is greater than
the total weight of all the other neighbors of $b_i^1$ combined. Hence, $b_i^1 = 1$.
Similarly, we can argue that $a_i^2 = 0$ and $b_i^2 = 1$ if $I_2(g_i) = 1$. 

Since $b_i^1 = 1$ and $g_i = 1$, vertex $c_i$ is biased towards $0$ with weight at least
$2\cdot2^i\cdot2^{100N}$, which is greater than
the total weight of all the other neighbors of $c_i^1$ combined. Hence, $c_i^1 = 0$.
We now focus on vertex $z_i^2$. Its neighbors are two vertices of weight $2^i \cdot 2^{100N}$
with constant value $0$, vertices $c_i^1$, $c_i^2$ , $y_i^3$ and a constant vertex $1$
with weights $2^{100N}\cdot(2^i-50)$ and some auxiliary vertices of negligible weight. 
If $c_i^1 = 0$, then $z_i^2$ is biased towards $1$ with weight at least 
$3\cdot2^i \cdot 2^{100N}$, which is greater than the weight of the remaining
neighbors combined. Hence, in local optimum, $z_i^2 = 1$. Hence, the claim has 
been proved in this case. If $I_2(g_i) = 1$, the proof is analogous.

Now suppose $I_1(g_i) = I_2(g_i) = 0$ and $g_i = 0$. Since $I_1(g_i) = 0$, vertex
$d_i^1$ is biased towards $1$ with weight at least $2\cdot2^{i+1} \cdot 2^{100N}$, 
which is greater that the weight of all its other neighbors combined. 
Hence, $d_i^1 = 1$. Similarly, we can prove that $d_i^2 = 1$. 
This means that vertex $b_i^3$ is biased towards $0$ with weight at least
$2\cdot(2^i + 10) \cdot 2^{100N}$, which is greater than the weight
of its other vertices combined. This implies that $b_i^3 = 0$.
By the same reasoning, $v_i = 0$. 
Since $b_i^3 = g_i = 0$, vertex $c_i^3$ is biased towards $1$ with weight
at least $2\cdot2^i \cdot 2^{100N}$, which is greater than the weight
of its other vertices combined. Hence, $c_i^3 = 1$. 
Now we focus on vertex $y_i^2$. Its neighbors are a vertex of weight $2^i \cdot 2^{100N}$
with constant value $1$, vertex $c_i^3$ with weight $2^i \cdot 2^{100N}$, 
$z_i^1$,a constant vertex $1$ both with weight $2^{100N}\cdot(2^i_20)$ ,$z_i^2$ and a constant $0$
with weight $2^{100N}\cdot{2^i-10}$ and some auxiliary vertices of negligible weight.
Hence, $y_i^2$ is biased towards $0$ with weight at least 
$2\cdot2^i \cdot 2^{100N}$, which is greater than the weight of the remaining
neighbors combined. Hence, $y_i^2 = 0$. 

We are now going to prove that if $y_i^2 = 0$, then $z_i^2 = 1$, which concludes
the proof for this case and is also the second claim of the lemma. 
We first notice that $z_i^2$ is never biased towards $0$ by the vertices of the 
NOR gadget. Hence, if the bias of the remaining vertices is towards $1$, then $z_i^2 = 1$
in local optimum. We notice that vertices $y_i^2$ and constant vertex $0$ bias our vertex
with weight $2\cdot(2^i-10)\cdot2^{100N}$, which is greater that any potential
bias by vertices $y_i^3$ and constant $1$ in the chain of total weight $2\cdot2^{100N}\cdot(2^i-50)$. Hence, $z_i^2 = 1$. 

It remains to prove the last claim of the Lemma. It suffices to show that when 
a $z_i$ in the chain is $1$, the next $y_{i+1}$ will be $0$ and the claim will follow
inductively. By a similar argument to the one used for the second claim, vertex
$y_{i+1}$ is not biased towards $1$ by any vertex in the NOR gadget. However, it
experiences bias towards $0$ from vertex $z_i$ and constant vertex $1$, which is greater
than any other potential bias from its other neighbors. Hence, $y_{i+1} = 0$ and
the claim follows.
\end{proof}

\begin{lemma}\label{l:SY_correction}
Suppose $z^1_i=0$ and $y^1_i=1$. If $g_i$ is correct then $z^2$ and $y^2$ are indifferent with respect to the other vertices of the gate $G_i$. If $g_i$ is incorrect then $g_i$ is indifferent with respect to the other vertices of the gate $G_i$, but gains the vertex $\rho$ of weight $2^{-500N}$.
\end{lemma}

\begin{proof}
Assume $g_i$ is corrrect. 

Assume at least one of $I_1(g_i),I_2(g_i)$ is equal to 1, say $I_1(g_i)=1$, hence at least one of $d_i^1,d_i^2$, assume $d_i^1$, is equal to 0. This is because otherwise it would experience bias from its neighbors $I_1(g_i)$,$d_i^2$ towards 0 which, along with the bias from the auxiliary vertex between $y_i^1$ and $d_i^1$, would dominate it towards 0. Therefore, $b_i^3$ experiences bias towards 1 from both $g_i$, which is correct, and $d_i^1$, which means, along with the bias from $y_i^1$, its equal to 1. Hence, $c_i^3$ must be equal to 0, since it is dominated by the bias from $b_i^3$, the constant vertex of 1 and its auxiliary bias from $z_i^1$. Hence, $c_i^3$ is 0 and $y_i^2$ is indifferent. 

Assume $I_1(g_i)=0,I_2(g_i)=0$. Hence, $d_i^1=1,d_i^2=1$, which means that $b_i^3=0$. Since $g$ is correct, it must be $g=1$, and therefore $c_i^3=0$, since it can take its preferred value of 0, towards which it is biased by $z_i^1$. Therefore, $y_i^2$ is indifferent to this gadget.

Moreover, since $g_i=0$ it must be that $c_i^1=c_i^2=1$ since they both have $g_i$ and a constant 0 vertex as their neighbors, which along with the bias from $y_i$, dominates their bias. Since $z_i^2$ neighbors with two 1 vertices and two 0 vertices it is indifferent with respect to this gadget.

In all cases, when $g_i$ is correct, $y_i^2$, $z_i^2$ are indifferent.

Assume $g_i$ is incorrect.

Assume that at least one of $I_1(g_i)$,$I_2(g_i)$ is equal to 1. Similarly to above, $d_i^1=0$. Since, $g_i$ is incorrect $c_i^3$ will be 0 due to the bias from $g_i$, the constant vertex 1 and $z_i^1$. Hence, $b_i^3=1$. The vertex $g_i$ is therefore indifferent with respect to this gadget.

Furthermore, since $I_1(g_i)=1$, $a_i^1=0$ and $b_i^1=1$. Since $g_i=1$ we have that $c_i^3=0$. Also, since $I_2(g_i)=0$, $a_i^2$ must take its preferred value of 1, and hence $b_i^2$ takes its preferred value of 0. Similarly, $c_i^2$ can also take its preferred value of 1. Overall, $g_i$ is connected to $b_i^1=1$,$c_i^1=0$,$b_i^2=0$,$c_i^2=1$ and hence is indifferent 

Assume both $I_1(g_i)=I_2(g_i)=0$. Then $d_i^1=d_i^2=1$, which means that $b_i^3=0$, and since $g_i=0$, we have that $c_i^3=1$. Hence, $g_i$ is indifferent with respect to this gadget. 

Because $I_1(g_i)=I_2(g_i)=0$, we have that $a_i^1=a_i^2=1$ since they can take their preferred values. Moreover, $b_i^1=b_i^2=0$ since they are biased to 0 by $z_i^2$. Given that $g_i=0$ it must be that both $c_i^1=c_i^2=1$. Therefore, $g_i$ indifferent in this case as well.

Since in all cases that $g_i$ is incorrect, it is indifferent with respect to this gadget, it will adhere to the bias that the auxiliary gadget connecting $a_i^1,a_i^2,g_i$ gives to $g_i$. If both $I_i$ is 1 then $a_i^1=0$ and if $I_i=0$ then $a_i^1=1$. In all cases, $a_i = \neg I_i$. Hence, the auxiliary gadget gives bias to $g_i$ towards 0, except when both $I_1(g_i)=I_2(g_i)=0$ in which case it biases $g_i$ towards 1. This means that $g_i$ has a $2^{-500N}$ bias towards its NOR value.  
\end{proof}

\begin{lemma}\label{l:control_bias}
If $Control=1$ then all $y,z$ vertices have a $2^{-87N}$ bias towards their natural values. If $Control=0$ then all $y,z$ vertices have a $2^{-87N}$ bias towards their unnatural values.
\end{lemma}

\begin{proof}
The $NotControl$ vertices are dominated by $Control$'s bias of $2^{7N}$ and hence have the opposite value. By Lemma~\ref{l:leverage} we have that $Control$ and $NotControl$ experience at most $2^{6N}$ bias, while the $y,z$ vertices experience $2^{-87N}$ bias towards the values opposite $Control$ and $NotControl$, which proves the claim. 
\end{proof}

\begin{lemma}\label{l:full_from_bias}
Assuming all vertices of the computing circuit gadget are in local optimum and have no external biases. If $Control=1$ then $\forall i$,$ z_i^1=0,y_i^1=1,z_i^2=0,y_i^2=1,z_i^3=0,y_i^3=1$. If $Control=0$ then $\forall i z_i^1=1,y_i^1=0,z_i^2=1,y_i^2=0,z_i^3=1,y_i^3=0$. 
\end{lemma}

\begin{proof}
If $Control=1$, consider the highest $k$ such that $y_k$ or $z_k$ have their unnatural values, i.e. $y_k=0$ or $z_k=1$. Since $\emph{Control} = 1$ all $y,z$ vertices experience a bias towards their unnatural biases by Lemma~\ref{l:control_bias}. Since the bias that biases them towards their unnatural values is greater than the weight of the internal vertices connected to $y^1,z^1,y^3,z^3$ it must be that one of the vertices $y^2,z^2$ have unnatural values. However, by Lemma~\ref{l:SY_detection} all control vertices for $j<k$ are also unnatural. Assume that the output vertices of $G_i$ are only internal to the circuit, i.e. no vertex except those belonging to the computing gadget is connected to them. Since by Lemma~\ref{l:SY_indifferency}, unnatural values for $y_i^1,z_i^1$ imply that the input vertices of gates $G_i$ are indifferent to $G_i$, the vertex $g_i$ would be dominated by the bias from the auxiliary vertex $\rho$. If $g_k$ is an output vertex by the assumption has no external bias. Hence, in this case, $g_k$ can take the correct value, which is a contradiction since $g_k$ having the correct value would mean $y_k^2,z_k^2$ can take their natural bias by Lemma~\ref{l:SY_correction}.

If $Control=0$, consider the least $k$ such that the $y,k$ control vertices have their natural values. Since the $y_i^1,z_i^1,y_i^3,z_i^3$ vertices are dominated by the $2^{-87N}$ bias ensured by Lemma~\ref{l:control_bias} we have that the only vertices with natural values can be $y_i^2,z_i^2$. However, even these vertices can only be biased towards their unnatural values since, by lemmas~\ref{l:SY_detection}~and~\ref{l:SY_correction}, even if the gate is correct $y_i^2,z_i^2$ are indifferent with respect to the NOR gadget.
\end{proof}

Having proved the above auxiliary lemmas, we can finally prove the theorem specifying the behaviour of our computing circuits.

\sygadgets*

\begin{proof}
\noindent
\begin{enumerate}
    \item Since $\emph{Control}\ell = 1$ and since we assumed no vertex experiences any external bias, by lemma \ref{l:full_from_bias} we have that all $y,z$ have their natural values and hence all gates compute correctly, by lemma \ref{l:SY_detection}. Therefore, $Next\ell = Real-Next(I_\ell)$ and $Val_\ell = Real-Val(I_\ell)$.

    \item Since $Control\ell=0$, by lemma \ref{l:full_from_bias} all $y,z$ have their unnatural values. Since all NOR gadgets have unnatural control vertices we have that their inputs are indifferent with respect to the gadgets. Hence, the claim that they are unbiased follows. 
    
    \item The $Control\ell$ vertex is connected to a vertex $NotControl\ell$, of weight $W_{NotControl\ell}=2^{7N}$, as well as to several leverage gadgets, which contribute bias at most $2^{100N-94N}=2^{6N}$. Hence, the $2^{7N}$ bias dominates.
\end{enumerate}
\end{proof}

%% file: Appendix_EqualityGadget.tex
\subsection{The Equality Gadget}\label{s:control}

The \emph{Equality} Gadgets are used to check whether the next best neighbor of a circuit has been 
successfully transferred to the input of the other circuit. 
The output of the \emph{Equality} gadget is connected to the control variables of the
circuit that should receive the new input. If the new input has not been transferred,
the output of this gadget biases the \emph{Control} vertex towards 0, which biases the internal control vertices towards unnatural values.
This enables the inputs of the circuit to change successfully to the next solution. 
When the new solution is transferred, the output of the gadget changes, in order to bias the
control vertices towards their natural values, so that the computation can take place.

Since we have two possible directions, both from Circuit $A$ to Circuit $B$ and vice versa we 
need two copies of the gadgets described in this section.

We will now describe the function of the \emph{Equality} Gadget when Circuit $A$
gives feedback to Circuit $B$. The \emph{Equality} Gadget takes as inputs the $T_A$ vertices from the \emph{CopyA} Gadgets and $I_B$ and simply checks whether they are equal.
Due to Lemma \ref{l:copy2}, at local optimum, the $T_A$ vertices have the same value
as \emph{NextA} which we want to transfer. One might try to connect 
\emph{NextA} as input to the \emph{Equality} gadget. The reason we avoid this construction is that
we do not want the output vertices of the \emph{Circuit Computing} gadget $C_A$ to experience any bias from this
gadget, because the computation changes their value with very small bias.
For this reason, we connect $T_A$ vertices to the input that are dominated by $\eta_A$
vertices. The input vertices $I_B$ are dominated by either the vertices in the NOR gadgets or
$\eta_A$, hence we can connect them directly as inputs to the gadget.

For each bit of the next best neighbor, we construct a gadget as in Figure \ref{fig:eqComp},
which performs the equality check for the $i$-th bit of the next best neighbor.
The idea for this construction is very simple: the weights decrease as we come
closer to the output, so that the input values dominate the final result. 
If the inputs are equal, the final value will be $0$. Notice  that we have put
and intermediate vertex between $I_B$ and the gadget to ensure that the two input
vertices will have equal weight. 
A detailed analysis is provided in the proof of Lemma \ref{l:consistent1}. 

\begin{figure}[t]
    \centering
    \includegraphics[scale=0.85]{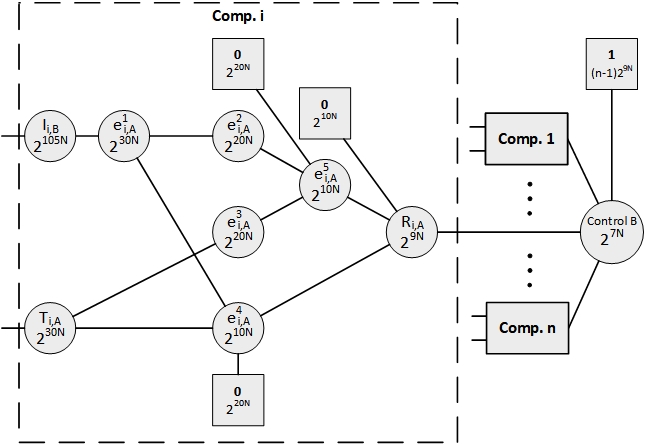}
    \caption{This gadget performs equality check for the bits $I_{i,B}$ and $T_{i,A}$. If they are
    equal, $R_{i,A} = 0$ at local optimum. We have $n$ such gadgets for each of the two circuits.
    The $n$ gadgets are connected to produce the final output, which is \emph{ControlB}. }
    \label{fig:eqComp}
\end{figure}

Now that we have gadgets to perform bit wise equality checks, we need to connect them
all to produce the output of the \emph{Equality} gadget. This is done by the construction of Figure~\ref{fig:eqComp}
Essentially, the idea is that if all the bits are equal, all the comparison results will
be $0$ and will dominate the \emph{ControlB} to take $1$. If at least one result is $1$, then together 
with the constant vertex $1$ will bias \emph{ControlB} to take value $0$.

We now prove the main lemma concerning the \emph{Equality} Gadget, which states that at local optimum, the output of the \emph{Equality} will be $1$ if and only if the two inputs to the 
gadget are equal. 

\consistent*

\begin{proof}
For simplicity we only prove the second claim, since the first follows by similar
arguments. We first focus on on the behavior of a single \emph{Equality} gadget. We would like
to prove that $R_{i,A} = 0$ if and only if $I_{i,B} = T_{i,A}$.

We first observe that vertex $e_{i,A}^1$ is biased with weight $2^{105N}$ by
$I_{i,B}$, which is greater that the bias from its other neighbor $e_{i,A}^2$. Hence, at
local optimum it is always the case that $e_{i,A}^1 = \neg I_{i,B}$. 
Moreover, vertices $e_{i,A}^2$ and $e_{i,A}^3$ essentially function as the complements of
$e_{i,A}^1$ and $T_{i,A}$. This is because they are biased with weight $2^{30N}$ by them,
which is greater than the bias by vertex $e_{i,A}^5$. Hence, $e_{i,A}^2 = \neg e_{i,A}^1$ 
and $e_{i,A}^3 = \neg T_{i,A}$.

We first examine the case where $I_B = T_A$. Then $e_{i,A}^1 = \neg T_{i,A}$. Since these vertices
have equal weights, vertex $e_{i,A}^4$ experiences $0$ total bias from them and is biased by
constant vertex $0$ with weight $2^{20N}$ and by $R_2$ with weight $2^{9N}$. Therefore,
$e_{i,A}^4 = 1$. By the previous observations we have that $e_{i,A}^2$ and $e_{i,A}^3$ have opposite 
values, which means that $e_{i,A}^5$ has bias $0$ from these two vertices. Is also has bias $2^{20N}$
by constant $0$ and $2^{9N}$ from $R_{i,A}$. Hence, $e_{i,A}^5 = 1$. Vertices $e_{i,A}^4$ and $e_{5.i}$ 
bias vertex $R_{i,A}$ towards $0$ with weight $2\cdot2^{20N}$, which is greater than the bias 
from constant $0$ and \emph{ControlB}. As a result, we have that $R_{i,A} = 0$ and the argument
is complete in this case.

Now we examine the case where $I_{i,B} \neq T_{i,A}$. Assume that $I_{i,B} = 1$, the other case
follows similarly. Then, $e_{i,A}^1 = 0$, $T_{i,A} = 0$, $e_{i,A}^2 = 1$, $e_{i,A}^3 = 1$. This means that $e_{i,A}^5$ is biased with
weight at least $2\cdot2^{30N}$ towards $1$, which is greater than the 
combined weight of $R_{i,A}$ and constant $0$. Therefore, $e_{i,A}^5 = 0$. Now we observe that
$R_{i,A}$ is biased with weight at least $2\cdot2^{20N}$ towards $1$ by vertices $e_{i,A}^5$ and 
constant $0$, which is greater that the combined weight of $e_{i,A}^4$ and \emph{ControlB}.
Hence, $R_{i,A} = 1$ in this case. If $I_{i,B} = 0$, then we could prove similarly that 
$e_{i,A}^4 = 0$, which implies that $R_{i,A} = 1$ by the same argument. 

We will now prove that \emph{ControlB} takes the appropriate value. 
First of all, we observe that \emph{ControlB} is connected with \emph{NotControlB}, (part of the \emph{Circuit Computing} gadget)
which has weight $2^{7N}$ and with $R_{i,A}$ vertices which have weight $2^{9N}$.
It is also biased with weight slightly more than $2^{6N}$ by each of the control variables
$y_i$ due to the leverage gadget. This means that for $N$ large enough \emph{ControlB} is dominated by the behavior of the $R_{i,A}$ vertices.
Suppose that $I_{i,B} = T_{i,A}$
for all $i$, $1 \leq i \leq n$. By the preceding calculations, we have that $R_{i,A} = 0$ for
all $i$. Hence, $ControlB$ experiences total bias $n\cdot2^{9N}$ towards $1$, which is greater
than the weight of constant vertex $1$. Thus, $\text{\emph{ControlB}} = 1$ in this case. 
Now suppose that there exists a $j$, $1\leq j \leq n$, such that $I_{2,j} = \neg T_{j,A}$. By the
preceding calculations, $R_{j,A} = 1$. Hence, vertex \emph{ControlB} is biased by vertices $R_{j,A}$ and 
constant $1$ towards $0$ with weight at least $(n-1)\cdot2^{9N} + 2^{9N} = n\cdot2^{9N}$, which is
greater than the combined weight of all the other $R_{i,A}$'s. Therefore, $\text{\emph{ControlB}} = 0$
in this case and the proof is complete. 
\end{proof}

%% file: Appendix_CopyGadget.tex
\subsection{The Copy Gadget}\label{s:copy}

The \emph{Copy} Gadgets transfer the values of the next best Neighbor of a circuit to
the input of the other circuit. This is fundamental for the correct computation
of the local optimum. There are some technical conditions that these gadgets should
satisfy, which we discuss in the following.

The purpose of the \emph{Copy} Gadgets is twofold. Firstly, when the \emph{Flag} vertex has value 1, 
they are meant to give the inputs of Circuit $B$ a slight bias to take the values
of the best flip neighbor that Circuit $A$ offers, that is \emph{NextA}. 
Secondly, in this case they are meant to give zero bias to the output vertices 
of Circuit $A$ that calculate the best flip neighbors. This is because when vertex
\emph{Flag} is $1$, the input of circuit $A$ is going to change, which means
that the NOR gates of this circuit will compute the new values. A consequence of the
functionality of the NOR gadgets is that the outputs of a gadget are only biased 
towards the correct value with a very small weight. This is because the gadget is
constructed in a way that allows these vertices to be indifferent to all of their neighbors
when the time comes to change their value. As a result, if we connect the output
vertices with other gadgets, we have to ensure that they will experience zero bias from
them in order for the computation to take place properly. Since the outputs of 
Circuit $A$ that produce the next best neighbor are connected to the \emph{Copy} Gadgets, 
we should ensure that they will experience zero bias when vertex \emph{Flag} is $1$, 
so that they can change properly. 
A similar functionality should be implemented when vertex \emph{Flag} is $0$. 

Next, we present the gadgets that implement the above functionality.
There are two \emph{Copy} gadgets with similar topology, \emph{CopyA} and \emph{CopyB}. For simplicity,
we only describe the details of \emph{CopyA}. The gadget takes as input the value of 
vertex \emph{Flag}, which determines whether a value should be copied or
whether the outputs of Circuit $A$ should experience zero bias. It also takes as 
input \emph{NextA}, which is the next best neighbor calculated by Circuit $A$. The 
output of the gadget is a bias to vertices $I_B$ and $T_A$ towards adopting the value of \emph{NextA}.

At this point, one might wonder why we didn't just connect the output of the \emph{CopyA}
gadget to the input $I_B$. This is because the value of $I_B$ also depends on the 
control variables. If the control variables of the input gates have natural values,
then the inputs experience great bias from the gate, making it impossible for
their values to change by the \emph{Copy} Gadget. Hence, the \emph{Copy} Gadget gives a slight
bias to vertex $T_A$, which is an input to an auxiliary circuit that compares it with $I_B$ (i.e the $Equality$ gadget)
. If they are not equal, this means that the output has not been transferred yet. In 
this case, the output of the gadget is given a suitable value to bias the
control vertices towards unnatural values. When this happens, the inputs $I_B$ 
can change to the appropriate values.

\begin{figure}[t]
    \centering
    \includegraphics[scale=0.8]{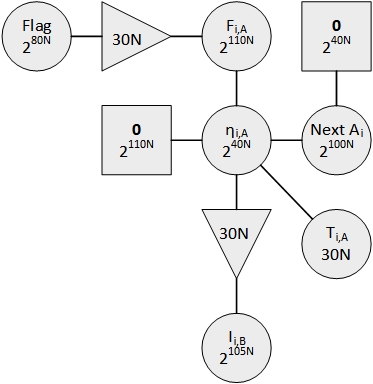}
    \,\,\,\,\,\,\,\,\,\,
    \includegraphics[scale=0.8]{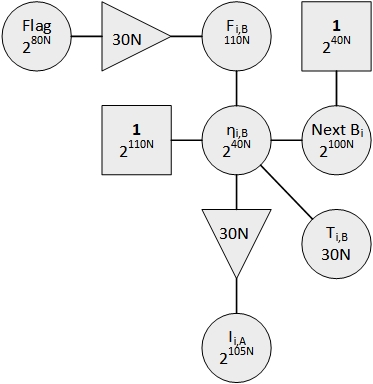}
    \caption{The gadgets that copy the values from one circuit to the other}
    \label{fig:feedbackA}
\end{figure}

Note that we have one of the above gadgets for each of the bits of the next best neighbor solution that the 
\emph{Circuit Computing} gadgets output.

We have a gadget of Figure \ref{fig:feedbackA} for each of the $m$ bits of the next best neighbor. Vertex $F_{i,A}$
has a very large weight in order to dominate the behavior of $\eta_{i,A}$. However, we do not want
this vertex to influence the behavior of \emph{Flag}. For this reason, we connect \emph{Flag} with $F_{i,A}$ using a \emph{Leveraging} 
gadget. 
Notice that the behavior of $F_{i,A}$ is dominated by \emph{Flag} by weight at least $2^{50N}$.
Another important point is that we connect the output of the \emph{CopyA} gadget with the
input of Circuit $B$ using another \emph{Leveraging} gadget. This is due to the fact that the 
weight of the input vertices is of the order of $2^{105N}$, which is far more than the
weight of $\eta_{i,A}$. Hence, we do not want the input vertices to influence the value of
$\eta_{i,A}$, while also ensuring that the \emph{Copy} gadget gives a slight bias to the inputs $I_B$
towards the value of \emph{NextA}.

We now prove Lemma~\ref{l:copy2}, which makes precise the already stated claims
about the function of the \emph{Copy} Gadgets.

\copyB*

\begin{proof}
We prove the claim for $\text{\emph{Flag}}=1$. The case $\text{\emph{Flag}}=0$ is identical.

We begin with the first claim. Due to the leveraging gadget, vertex $F_{i,B}$ experiences 
bias from \emph{Flag} which is slightly less than $2^{50N}$. Hence, it is biased towards
$0$ with weight at least $2^{49N}$. This is greater than the weight of $\eta_{i,B}$, 
which is the other neighbor of $F_{i,B}$. Hence, $F_{i,B} = 0$ at local optimum. 
Now vertex $\eta_{i,B}$ experiences zero total bias from vertices $F_{i,B}$ and constant $1$ and 
biases $2^{100N}$ by $\text{\emph{NextB}}_i$, $2^{30N}$ by $T_{i,B}$ and slightly more that $2^{75N}$ by 
the input $I_{i,A}$ due to leveraging,
which means that its value at local optimum will be determined by $\text{\emph{NextB}}_i$.
Specifically, $\eta_{i,B} = \neg \text{\emph{NextB}}_i$ at local optimum. Now, vertex $T_{i,B}$ experiences bias
$2^{40N}$ from $\eta_{i,B}$ and biases of the order of $2^{7N}$ from the gates of the controller
gadget. Hence, $T_{i,B}$ has bias towards $\text{\emph{NextB}}_i$ equal to $w_{\eta_{i,B}}$ and will take
this value at local optimum.

To prove the second claim, we use the already proven fact that $\eta_{i,B} = \neg \text{\emph{NextB}}_i$ when
$\text{\emph{Flag}} = 1$. Due to the \emph{Leverage} gadget, vertex $\text{\emph{I}}_i$ experiences bias slightly less than 
$2^{10N}$ from vertex $\eta_{i,B}$. Since $\text{\emph{ControlA}}=0$, by Lemma~\ref{l:SY_gadgets}, we have that $I_{i,A}$ is indifferent with respect to the gadget $C_A$, and will therefore take the value of $\neg \eta_{i,B} =  \text{\emph{NextB}}_i = T_{i,B}$
\end{proof}

\copyUnbias*
\begin{proof}
We notice that due to leveraging, vertex $F_{i,A}$ of gadget \emph{CopyA}
experiences bias slightly less than $2^{50N}$ from vertex $\text{\emph{Flag}} = 1$. This dominates its
behavior, since the other neighbor $\eta_{i,A}$ has weight that is orders of magnitude smaller.
Hence, $F_{i,A} = 0$. Now, vertex $\eta_{i,A}$ experiences total bias $2*2^{110N}$
from vertices $F_{i,A}$ and constant $0$, $2^{100N}$ from $NextA_i$, $2^{30N}$ from $T_{i,A}$ 
and slightly more than $2^{75N}$ from $I_{i,B}$ due to the \emph{Leverage} gadget used. 
This means that $\eta_{i,A} =1$. Now we are ready to prove our claim. Vertex $\text{\emph{NextA}}_i$
is connected to vertices $\eta_{i,A}$ and constant $0$ of gadget $\text{\emph{CopyA}}_i$. They have the same
weight and opposite values at local optimum. This means that $\text{\emph{NextA}}_i$ has $0$ bias with respect to
$\text{\emph{CopyA}}_i$, i.e it is indifferent.

The case for $\text{\emph{Flag}}=0$ follows symmetrically.
\end{proof}

%% file: Appendix_ComparisonGadget.tex
\subsection{The Comparator Gadget}\label{s:comparison_gadget}

The purpose of the \emph{Comparator} gadget is to implement the binary comparison between the bits of the values of the two circuits. At the same time we need to ensure that the
vertices of the losing circuit (i.e the circuit with the lower value) are indifferent with respect to the \emph{Comparator} gadget, so that Lemma \ref{l:SY_gadgets} can be applied. 

In particular, the output vertices that correspond to the bits of the value, presented in section \ref{s:computing_gadgets}, with weights $2^{90N}\cdot 2^{N+i}$ are each connected as below.

Note that the output bits of the second circuit B are the complement of their true values, in order to achieve comparison with a single bit. The weight of the \emph{Flag} vertex is $2^{80N}$

\begin{figure}[t]
    \centering
    \includegraphics[scale=0.8]{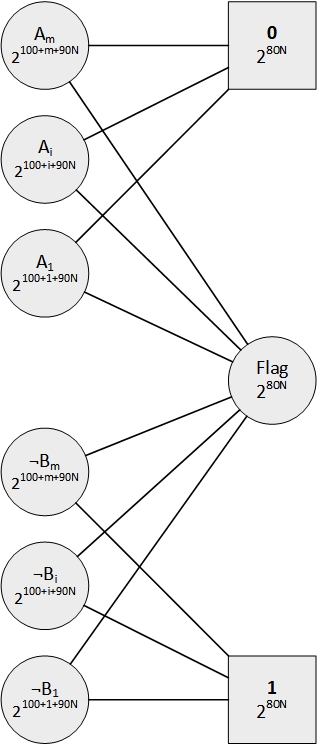}
    \caption{vertices of the \emph{Comparator} gadget. Note that Circuit B is meant to output the complement of its true output.}
\label{fig:comparator}
\end{figure}

To see why the value of vertex \emph{Flag} implements binary comparison one needs to consider four cases: In the first two, where the $i$th bits are both equal, the total bias \emph{Flag} experiences is zero, since it experiences bias towards a certain bit as well as the complement of said bit. In the other two, where one bit is 1 and the other is 0, the \emph{Flag} vertex will experience $2^i$ bias towards either value, which will supersede all lower bits.

However, the \emph{Comparator} gadget is meant not only to implement comparison between 
values, but also to detect whether a circuit is computing wrongly and, hence, to fix it.
To this end we connect the following control vertices to the vertex \emph{Flag}: the control 
vertices $y^3_{m+1,A}$ for circuit $A$ and $z^3_{m+1,B}$ for circuit B, where $m+1$ is the last NOR gadget 
before the bits of the values (recall that we have $m$ value bits and that $w_{y^3_{i,A}}=w_{z^3_{i,B}}=2^{100N}\cdot(2^{N+m+1}-50)$) 
(see Figure~\ref{f:chain}), as well as the control vertices $y^3_{i,A},z^3_{i,B},\forall i \leq m$
for each NOR gadget that corresponds to an output bit of the value (which have weight 
$w_{y^3_{i,A}}=w_{z^3_{i,B}}=2^{90N}\cdot(2^{N+i}-50)$). The vertices $y^3_{m+1,A}$ and $z^3_{m+1,B}$ are used to check whether 
the next best neighbor has been correctly computed. If it isn't, these vertices dominate \emph{Flag}, 
due to their large weight of $2^{100N}$ compared to the weight of the value bits, which is 
of the order of $2^{90N}$. The control vertices of the output bits of the value are used in 
a more intricate way to ensure that even if one to the results is not correct, the 
output of the comparison is the desired one. Details are provided in Lemma~\ref{l:super_comparison}. 
All these vertices are connected in such a way that a control vertex with unnatural value, biases \emph{Flag} 
towards fixing that circuit.

\begin{figure}[t]
    \centering
    \includegraphics[scale=0.5]{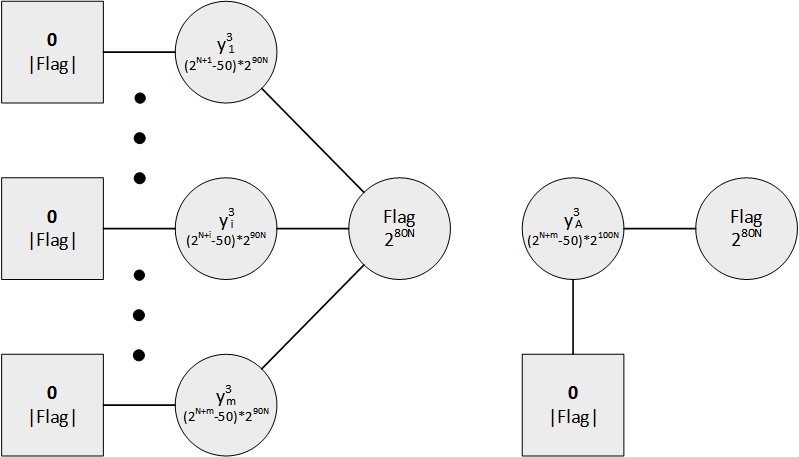}
    \includegraphics[scale=0.5]{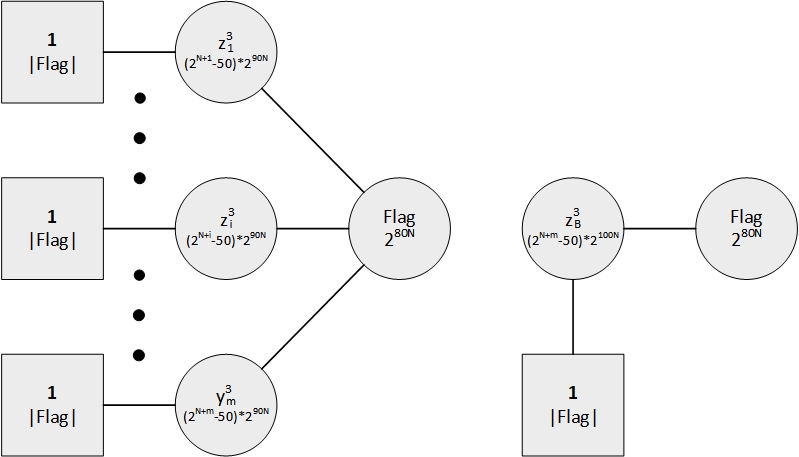}
    \caption{Connection between the \emph{control vertices} and the \emph{Flag} vertex.}
\label{fig:detection}
\end{figure}

We prove the following properties:

\comparatorUnbias*

\begin{proof}
Suppose $\text{\emph{Flag}}=1$. Then the only vertices of $C_A$ connected to the \emph{Comparator} gadget are the value output bits and certain control vertices, in such a way that they are connected to either \emph{Flag} or a constant vertex $0$ of weight equal to \emph{Flag}. In all cases, both biases  
cancel each other out and the vertices of Circuit $A$ are indifferent.
Suppose $\text{\emph{Flag}}=0$. Then the only vertices of $C_B$ connected to \emph{Flag} are also connected with a constant $1$ vertex. Similarly to the first case, all vertices of circuit B are indifferent with respect to the \emph{Comparator} gadget when $\text{\emph{Flag}}=0$.
\end{proof}

We now prove the most important lemma of the \emph{Comparator} gadget.
Our goal is to compare the output values of the two circuits, so that we change the
input of the circuit with the smaller real value. The main difficulty lies in that
one or both of the circuits might produce incorrect bits in their output. 
A simple idea would be to try to detect any incorrect output bits and 
influence \emph{Flag} accordingly, as we do with control variables $y^3_{m+1,A}$ and $z^3_{m+1,B}$.
However, if the least significant bit of a circuit is incorrect, the weight of the corresponding
control vertex is exponentially smaller that the rest of the bits. Hence, it cannot
dominate the outcome of the comparison. This means that sometimes we might be
at a local optimum where some output vertices are incorrect. To alleviate this problem we propose this construction.

The idea behind this lemma is very simple: if it is guaranteed that the output of
one of the circuits is correct and we know which bits of the other circuit \emph{might}
be wrong, we can still compare their true values.
This is accomplished by an extension of the traditional comparison method, by also taking
into account the control variables of the output bits and examining all the possible
cases. 
This lemma is very useful in our proof, since by Lemma~\ref{l:correctness1} we know that at least one of the circuits
computes correctly in local optimum. 

\begin{lemma}\label{l:super_comparison}
At any local optimum:

Suppose that $\text{\emph{Flag}}=1$. If $\forall i, z^1_{i,A}=0, y^1_{i,A}=1,z^2_{i,A}=0, y^2_{i,A}=1, z^3_{i,A}=0, y^3_{i,A}=1$ and $\forall i>m, z^1_{i,B}=0, y^1_{i,B}=1,z^2_{i,B}=0, y^2_{i,B}=1, z^3_{i,B}=0, y^3_{i,B}=1$ then 
\[ Real\text{-}Val(I_A) \leq Real\text{-}Val(I_B) \]

Suppose that $\text{\emph{Flag}}=0$. If $\forall i, z^1_{i,B}=0, y^1_{i,B}=1,z^2_{i,B}=0, y^2_{i,B}=1, z^3_{i,B}=0, y^3_{i,B}=1$ and $\forall i>m, z^1_{i,A}=0, y^1_{i,A}=1,z^2_{i,A}=0, y^2_{i,A}=1, z^3_{i,A}=0, y^3_{i,A}=1$ then 
\[ Real\text{-}Val(I_B) \leq Real\text{-}Val(I_A) \]
\end{lemma}

\begin{proof}
Since for all gates that do not correspond to value bits (see Figure~\ref{f:chain}), we have that they possess natural values, and hence \emph{Flag} is indifferent with respect to them, we only need to examine the final $m$ gates that correspond to the value bits. 

We denote the $k$th bit of $\text{\emph{Val}}_A,\text{\emph{Val}}_B$ as $A_k,B_k$. $B_k$ corresponds to the actual value of the $k$th bit of the circuit B instead of its complement for simplicity. The actual value of the vertex corresponding to $B_k$ is the opposite. We also denote $z^2_{k,B}$ the control vertex corresponding to the bit $B_k$. We make the distinction between $A_k,B_k$ and $Real(A_k),Real(B_k)$. These may be equal or different depending on whether the circuit calculated the $k$th bit correctly. By the assumption we know that $A_k=Real(A_k)$ since $A$ calculates correctly. We do not know whether $B_k=Real(B_k)$, but we do know that $B_k \neq Real(B_k) \implies z^2_{k,B}=1$.

We consider three cases.

In the case that $(A_k,B_k,z^2_{k,B}) \in {(0,0,0),(1,1,0),(0,1,1)}$, \emph{Flag} experiences bias at most $2\cdot2^{90N}\cdot(50)$ from this bit towards $\text{\emph{Flag}}=1$ in any of these cases. In this case, we have that either $Real(A_k)=Real(B_k)$ or $Real(A_k)<Real(B_k)$, depending on whether $B_k$ calculated correctly. Either way, $Real(A_k) \leq Real(B_k)$. 

In the case that $(A_k,B_k,z^2_{k,B}) \in {(0,1,0)}$, \emph{Flag} experiences bias $2\cdot2^{90N}\cdot(2^{N+k})$ from this bit towards $\text{\emph{Flag}}=1$. In this case, we have that $Real(A_k) < Real(B_k)$, since both calculate correctly.

In the case that $(A_k,B_k,z^2_{k,B}) \in {(0,0,1),(1,0,0),(1,1,1),(1,0,1)}$ then \emph{Flag} experiences bias at least $2\cdot2^{90N}\cdot(2^{N+k}-50)$ towards $\text{\emph{Flag}}=0$ from this bit in any of these cases. In these cases, $Real(A_k)$ might be higher, but we will show that these cases can never matter.

Suppose $k$ the highest $i$ for which $(A_i,B_i,z^2_{k,B}) \notin {(0,0,0),(1,1,0),(0,1,1)}$. 

If no such $k$ exists then all bits must lie in the first case and hence $\forall k Real(A_k) \leq Real(B_k)$. Hence, $Real\text{-}Val(I_A) \leq Real\text{-}Val(I_B)$.

If for that $k$, $(A_k,B_k,z^2_{k,B}) \in {(0,1,0)}$, we know that $Real(A_k) < Real(B_k)$ while for all higher bits $k Real(A_k) \leq Real(B_k)$. This means that $Real\text{-}Val(I_A) < Real\text{-}Val(I_B)$, since the lower bits don't matter as long as we have a strict inequality in a high bit.

Lastly, if we have that $(A_k,B_k,z^2_{k,B}) \in {(0,0,1),(1,0,0),(1,1,1),(1,0,1)}$, we have that \emph{Flag} experiences bias at least $2\cdot2^{90N}\cdot(2^{N+k}-50)$ towards $\text{\emph{Flag}}=0$ from this bit. Furthermore, it experiences bias at most $2\cdot2^{90N}\cdot(50)$ towards $\text{\emph{Flag}}=1$ from each bit higher that $k$. Each bit $i$ lower than $k$ causes bias at most $2\cdot2^{90N}\cdot(2^{N+i})$ each towards $\text{\emph{Flag}}=1$. In total, if we have $m$ bits, we have at most $(m-k)\cdot2\cdot2^{90N}\cdot(50)+\sum_{i<k} 2\cdot2^{90N}\cdot(2^{N+i}) \leq (m)\cdot2\cdot2^{90N}\cdot(50)+2\cdot2^{90N}\cdot(2^{N+k}-2^N)$ towards $\text{\emph{Flag}}=1$ and at least $2\cdot2^{90N}\cdot(2^{N+k}-50)$ towards $\text{\emph{Flag}}=0$. For N sufficiently high, the bias towards $0$ would win, making \emph{Flag} no longer have 1 as its best response, which is a contradiction. Hence, the third case can not happen in a local optimum with $\text{\emph{Flag}}=1$.

The case for $\text{\emph{Flag}}=0$ is identical, with the only difference being we consider $y^2_{k,A}$ instead.
\end{proof}

\comparatorCorrectness*

\begin{proof} 
Assume a local optimum with $\text{\emph{Flag}}=1$ and $NextB \neq Real\text{-}Next(I_B)$. By Lemma~\ref{l:SY_detection} we have that Circuit $B$ is computing incorrectly and hence the control vertex $z^3_{m+1,B}$ (i.e. the last gate before the value bits) has its unnatural value, which is $z^3_{m+1,B}=1$. 

Assume that, $ControlA=0$. Then by Lemma~\ref{l:copy2} we have that $I_A=T_B$, which by Lemma~\ref{l:consistent1} we have $ControlA=1$, a contradiction. Hence, $ControlA=1$.

Therefore, since have that $ControlA=1$ and that $NextA = Real\text{-}Next(I_A)$, which by Lemma~\ref{l:full_from_bias}, implies that the corresponding vertex $y^3_{m+1,A}$ has its natural value $y^3_{m+1,A}=1$. 

This means that \emph{Flag} experiences bias towards 0 at least $2\cdot2^{100N}\cdot(2^{N+m+1}-50)$ from the vertices $z^3_{m+1,B}$,$y^3_{m+1,A}$, which dominates \emph{Flag} to take value 0. This is a contradiction since we assumed that $\text{\emph{Flag}}=1$ at local optimum. Hence, if $\text{\emph{Flag}}=1$ then it must be that $NextB = Real\text{-}Next(I_B)$.

Similarly, we can prove that if $\text{\emph{Flag}}=0$ then $NextA = Real\text{-}Next(I_A)$.
\end{proof} 

\comparator*

\begin{proof}

Assume that, $ControlA=0$. Then by Lemma~\ref{l:copy2} we have that $I_A=T_B$, which by Lemma~\ref{l:consistent1} we have $ControlA=1$, a contradiction. Hence, $ControlA=1$.

Since $\text{\emph{Flag}}=1$ and we have that $ControlA=1$, by Lemma~\ref{l:full_from_bias} all control vertices of $C_A$ have their natural values. Furthermore, since by the proof of Lemma~\ref{l:comparator_correctness} we know that all control vertices of weight $2^{100N}$ have their natural values, we can apply Lemma~\ref{l:super_comparison}. Therefore, $Real\text{-}Val(I_A) \leq Real\text{-}Val(I_B)$
   
The proof for $\text{\emph{Flag}}=0$ is identical.
\end{proof}